\documentclass[11pt]{article}

\usepackage{geometry}
\usepackage{enumitem}

\usepackage{amssymb}
\usepackage{mathrsfs}
\usepackage{amsmath}

\newtheorem{theorem}{Theorem}[section]

\newtheorem{lemma}[theorem]{Lemma}
\newtheorem{proposition}[theorem]{Proposition}
\newtheorem{claim}[theorem]{Claim}
\newtheorem{definition}[theorem]{Definition}

\def\squarebox#1{\hbox to #1{\hfill\vbox to #1{\vfill}}}
\newcommand{\qed}{\hspace*{\fill}
\vbox{\hrule\hbox{\vrule\squarebox{.667em}\vrule}\hrule}\smallskip}

\newenvironment{proof}{\noindent{\bf Proof:~~}}{\(\qed\)}

\begin{document}

\title{Combinatorial Cost Sharing}



\author{
	Shahar Dobzinski
	\and 
	Shahar Ovadia\thanks{Weizmann Institute of Science. The first author is the incumbent of the Lilian and George Lyttle Career Development Chair. Work supported in part by the I-CORE program of the planning and budgeting committee and the Israel Science Foundation 4/11 and by EU CIG grant 618128. Emails: \texttt{shahar.dobzinski@weizmann.ac.il, shaharov59@gmail.com}.  } }
\maketitle

\begin{abstract}
We introduce a combinatorial variant of the cost sharing problem: several services can be provided to each player and each player values every combination of services differently. A publicly known cost function specifies the cost of providing every possible combination of services. A combinatorial cost sharing mechanism is a protocol that decides which services each player gets and at what price. We look for dominant strategy mechanisms that are (economically) efficient and cover the cost, ideally without overcharging (i.e., budget balanced). Note that unlike the standard cost sharing setting, combinatorial cost sharing is a multi-parameter domain. This makes designing dominant strategy mechanisms with good guarantees a challenging task.

We present the Potential Mechanism -- a combination of the VCG mechanism and a well-known tool from the theory of cooperative games: Hart and Mas-Colell's potential function. The potential mechanism is a dominant strategy mechanism that always covers the incurred cost. When the cost function is subadditive the same mechanism is also approximately efficient. Our main technical contribution shows that when the cost function is submodular the potential mechanism is approximately budget balanced in three settings: supermodular valuations, symmetric cost function and general symmetric valuations, and two players with general valuations.
\end{abstract}

\section{Introduction}

In the cost sharing problem the cost of a public good (e.g., a bridge, a park, or a networking infrastructure) has to be partitioned among the players. The challenge is to decide whether to provision the public good, and if so, which players will receive usage permissions and for what prices, all while making sure that the payments cover the cost. 

The problem was extensively studied in both economics and algorithmic game theory and many variants were suggested and analyzed (e.g., \cite{littlechild1973simple,moulin1999incremental,moulin2001strategyproof,roughgarden2009quantifying,bleischwitz2007or,bleischwitz2008group,deb1999voluntary,mehta2007beyond,hashimoto2015strategy}). However, in almost all variants only a single good is considered.

In this paper we attempt to fill this lacuna and introduce \emph{combinatorial cost sharing}, where multiple goods can be provisioned and both costs and preferences depend on the selected combination of goods. While combinatorial cost sharing is a natural generalization of the basic cost sharing scenario, from a technical perspective it is radically different as we leave the relatively safe single parameter world and cross the bridge to the realm of multi parameter mechanism design. Nevertheless, we will see that good mechanisms for combinatorial cost sharing do exist.

\subsection*{Simple Cost Sharing}

The standard cost sharing setting (from now on, ``simple cost sharing'') involves a set $N$ of players ($|N|=n$), where the value of player $i$ is $v_i$ if player $i$ receives a usage permission and $0$ otherwise. A known cost function $C:2^N\rightarrow \mathbb R^+$ specifies the cost of serving each subset of the players. A (direct) mechanism for this problem receives as input the valuations of the players and outputs the set of served players and the payment $p_i$ of every player $i$. It is standard to assume that the mechanism is individually rational (for every $i$, $p_i\leq v_i$), that $p_i\geq 0$ (no positive transfers) and moreover $p_i=0$ if player $i$ is not served.

Work on cost sharing in the AGT community mostly focuses on incentive compatible mechanisms, either dominant strategy or groupstrategyproof, that at the very least always cover the cost. That is, in an instance where the mechanism serves the set $ALG$ of players, $C(ALG)\leq \Sigma_ip_i$. Ideally, we will also not overcharge the players, at least not by much: a mechanism is $\beta$-budget balanced if in every instance $C(ALG)\leq \Sigma_ip_i\leq \beta \cdot C(ALG)$.

We look for mechanisms that are economically efficient. Economic efficiency can be interpreted in several ways, with the \emph{social welfare} being the standard definition: $SW(S)=\Sigma_{i\in S}v_i-C(S)$. Unfortunately, Feigenbaum et al. \cite{feigenbaum2002hardness} show that no dominant strategy and budget balanced mechanism provides a finite approximation to the social welfare, and this result holds even if we only require the mechanism to cover the cost\footnote{This impossibility is typical for mixed sign objectives, like the social welfare. The simple proof is instructive: consider two players, each with $v_i=1$. The cost of serving any non-empty set is $1$. The social welfare of the optimal allocation (which is to serve both players) is $1$, whereas the social welfare of any other allocation is $0$. Therefore, in any mechanism that provides a finite approximation to the welfare both players are served. To cover the cost, one of the players, without loss of generality player $1$, must pay $p>0$. Consider now an instance with $v_1=\frac p 2>0$ and $v_2=1$. The mechanism is incentive compatible so player $1$ is not served. Hence, no matter whether the mechanism serves player $2$ or none of the players, the social welfare is $0$. However, the optimal social welfare is $v_1>0$.}.

To overcome this, Roughgarden and Sundararajan \cite{roughgarden2009quantifying} suggest an alternative quantification of efficiency. First, one can use additive approximations, i.e., if $OPT$ is the allocation that maximizes the welfare, then in every instance $SW(OPT)-SW(ALG)\leq \alpha\cdot C(OPT)$, for some reasonable $\alpha>0$. This benchmark allows the designer to shift the focus from instances with ``low'' welfare -- which are the raison d'etre of the impossibilities -- to instances in which using a cost sharing mechanism should yield a noticeable improvement in the efficiency.

A related notion discussed in \cite{roughgarden2009quantifying} is minimizing the \emph{social cost}, $\pi(S)=C(S)+\Sigma_{i\notin S}v_i$, which is the construction cost plus the ``lost value'' from not serving some of the players. Interestingly, the social cost and the social welfare induce the same order on the allocations (i.e., if $SW(S)\geq SW(T)$ then $\pi(S)\leq \pi(T)$). Moreover, additive approximations to the social welfare imply multiplicative guarantees on the social cost: $SW(OPT)-SW(ALG)\leq \alpha\cdot C(OPT)$ implies $\pi(ALG)\leq (\alpha + 1)\cdot \pi(OPT)$. Following \cite{roughgarden2009quantifying}, many papers study social cost minimization in various settings (e.g., \cite{chawla2006optimal,mehta2007beyond,bleischwitz2007or,bleischwitz2008group,brenner2007cost,gupta2015efficient}. For the sake of compatibility with the literature, we state our results in terms of approximation to the social cost, but in fact we prove essentially the same guarantees with respect to the stronger notion of additive approximation.

The cost sharing literature is rich in beautiful results, but the jewel in the crown is probably the Shapley value mechanism \cite{moulin2001strategyproof}, which is a groupstrategyproof mechanism that exactly shares the cost whenever $C$ is a submodular function \cite{moulin2001strategyproof}. Roughgarden and Sundararajan \cite{roughgarden2009quantifying} show that it gives an approximation ratio of $ \mathcal H_n=\Sigma_{i=1}^n\frac 1 n$ to the social cost. It is known that the approximation ratio of any mechanism that always covers the cost is $\Omega(\log n)$ and that this is true even if we relax the incentive compatibility requirement to strategyproofness \cite{dobzinski2008shapley}\footnote{This impossibility of \cite{dobzinski2008shapley} is stated for budget balanced mechanisms, but the proof applies even to cost recovering mechanisms.}.

\subsection*{The Model}

The main purpose of this paper is to go beyond the single good setting. Toward this end, we introduce \emph{combinatorial cost sharing}. We present two formulations of combinatorial cost sharing. The first is a more direct formulation which might help the reader to digest the setting more easily. The second formulation -- which is the one that is studied throughout the paper -- is equivalent in power but is notationally more involved. We use it since it makes the technical proofs more readable.

\paragraph{A First Attempt.} As in simple cost sharing, there is a set $N$ which consists of $n$ players, but now there is a set $M$ of public goods (for example, a pool, a gym, etc.). The mechanism has to decide which goods to construct and if a good is constructed, which subset of the players will be allowed to use it. Players might have complicated preferences over the goods in $M$ (e.g., a combined membership for the pool and the gym might be more valuable than the sum of the values of each membership alone), thus the private valuation of player $i$ is $v_i:2^{M}\rightarrow \mathbb R$. Note that this assumes that there are no externalities in the sense that value of each player is determined only by the goods he is served (in particular, the value does not depend on the other players who use these services).

Let $C':(2^N)^m\rightarrow \mathbb R$ be a known function that specifies the cost of every possible combination of services. For example, $C(S_1,\ldots, S_n)$ is the cost of serving the first good in $M$ to the players in $S_1$ while serving the second good in $M$ to the players in $S_2$, and so on. We stress that we do not make any assumptions on $S_1,\ldots, S_n$ and in particular these sets are typically not disjoint.

\paragraph{The Main Formulation.} The issue with the first formulation is that we often would like to assume that the cost function belongs to some standard class, e.g., $C'$ is submodular or subadditive. However, as defined $C'$ is not even a set function (its domain is $(2^N)^m$). We thus use a different formulation that is equivalent in power. Define -- for notational convenience -- for each player $i$ a set $M_i$ with $M_i\cap M_{i'}=\emptyset$ for $i\neq i'$, where each $j\in M_i$ represents a permission to consume a different good. For example, if $M$ is the set of public goods that can be constructed, we define for each player $i$ a set $M_i$, $|M_i|=|M|$, and think about the $j$'th item in $M_i$ as permission \emph{for player $i$} to use the $j$'th public good in $G$. In particular player $i$ is never interested in items from $M_{i'}$, for $i\neq i'$. The private valuation of player $i$ is $v_i:2^{M_i}\rightarrow \mathbb R$. The cost function $C:2^{M_1}\times\cdots\times 2^{M_n}\rightarrow \mathbb R$ specifies the cost of every possible combination of services. Note that it is straightforward to express every cost function in the first formulation as a cost function in the main formulation. In particular, $C$ is now a set function (the set of items is $M_1\cup \cdots\cup M_n$ -- recall that $M_i\cap M_j$ for $i\neq j$), so standard notions such as subadditivity and submodularity are defined in the usual sense.


\paragraph{Other Requirements.} Due to the combinatorial richness of our domain, we are not able to develop good \emph{group}strategyproof mechanisms. Instead, we focus on dominant strategy ones. Following the literature, we focus on mechanisms that satisfy individual rationality, no positive transfers and always cover the cost or maybe are even $\beta$-budget balanced, for some reasonable $\beta$. We naturally extend the social cost definition of \cite{roughgarden2009quantifying} to multiple goods: $\pi(\vec{S})=C(\vec{S})+\sum_{i\in N}\left[v_{i}(M_{i})-v_{i}(S_{i})\right]$. I.e., we still want to minimize the construction cost plus the lost value. We note that the extended definition preserves the properties discussed above, e.g., additive approximations to the welfare lead to multiplicative approximations to the social cost. We refer the reader to the preliminaries section for exact definitions.


\subsection*{Cost Recovering Mechanisms and the Potential Mechanism}

The simple cost sharing domain is a single parameter one, where the private information of every player consists of one number. Thus, to design a dominant strategy mechanism one can focus on the quite powerful family of monotone algorithms. In fact, the literature contains powerful techniques for designing groupstrategyproof mechanisms, for various notions of groupstrategyproofness (e.g., the Moulin family of mechanisms \cite{moulin1999incremental} and acyclic mechanisms \cite{mehta2007beyond}).

In contrast, the combinatorial cost sharing domain is a multi-parameter one. The difficulty of designing useful mechanisms for multi-parameter domains is well known. The root of evil is the lack of general design techniques except the VCG family. For example, if the domain is unrestricted, then the only possible dominant strategy mechanisms are affine maximizers \cite{roberts1979aggregation}, a slight variation of VCG mechanisms. In general, more restricted domains as ours do exhibit non VCG mechanisms, but the VCG family remains the main tool at our disposal.

However, while VCG is effective for welfare maximization, in cost sharing settings we also need to cover the construction cost. Unfortunately, conventional wisdom has it that the revenue of the VCG mechanism is uncontrollable and tends to be low\footnote{Some papers attempt to control the revenue of VCG in simpler auction settings by rebating the players, e.g., \cite{Moulin200996, guo2009worst, cavallo2006optimal}. Also relevant is the work of Blumrosen and Dobzinski \cite{blumrosen2014reallocation} which is the closest in spirit to ours (and in fact is the inspiration to our work). One of their results essentially provide a cost-recovering VCG based mechanism for the excludable public good problem. Their mechanism can be derived as a special case of our constructions.} \cite{ausubel2006lovely}. The main technical contribution of this paper challenges this -- we do manage to ``tame'' the VCG beast and obtain VCG based mechanisms that are approximately budget balanced.

For simplicity, we start our journey by constructing VCG based mechanisms that always cover the cost for simple cost sharing, so the valuation of each player $i$ can be described by a single number: $v_i$ if served and $0$ otherwise. In general, affine maximizers can have both (multiplicative) player weights and (additive) allocations weights. The former does not seem to be very useful, so we focus on affine maximizers of the form:
\begin{equation}
\arg\max_{S\subseteq N}\sum_{i\in S}v_{i}-H(S)\label{eq:VCG-Intro}
\end{equation}
where $H:2^N\rightarrow \mathbb R$ is a function that does not depend on the $v_{i}$'s. If $S$ is the allocation that maximizes (\ref{eq:VCG-Intro}) for the valuation profile $v$, then the payment of player $i$ is:
\begin{equation}
p_{i}=\sum_{j\in S_{-i}}v_{j}-H(S_{-i})-\left[\sum_{j\in S-\{i\}}v_{j}-H(S)\right]\label{eq:VCG-Payments-Intro}
\end{equation}
where $S_{-i}$ is the allocation that maximizes (\ref{eq:VCG-Intro})
when the valuation of player $i$ is identically $0$.

When $H$ is the cost function $C$, we get a welfare maximizing mechanism. However, it is common for this mechanism to run a deficit, e.g., in the special case of excludable public good ($C(S)=1$ for every $S\neq \emptyset$), if $v_{i}>1$ for every player $i$, the revenue is $0$. 

Thus, to cover the incurred cost we need some other function $H\neq C$. Notice that from the definition of $S_{-i}$ that $\sum_{j\in S_{-i}}v_{j}-H(S_{-i})\geq \sum_{j\in S-\{i\}}v_{j}-H(S-\{i\})$. Hence, the payment of the $i$'th player $p_{i}$ is at least $H(S)-H(S-\{i\})$. A function $H$ with the property that for every set $S\subseteq N$ it holds that $\sum_{i\in S}H(S)-H(S-\{i\})\ge C(S)$ will lead to a dominant-strategy mechanism for simple cost sharing that always collects payments that cover the incurred cost. For example, in the special case of excludable public good, we can choose $H(S)=\mathcal H_{|S|}$, so if a set $S$ is selected the marginal cost to $H$ of each player $i\in S$ is at least $H(S)-H(S-\{i\})= \frac 1 {|S|}$ and the total payment is at least $1$ (this special case was analyzed by Blumrosen and Dobzinski \cite{blumrosen2014reallocation}).

Interestingly, for \emph{every} cost function $C$ the \emph{potential function} of Hart and Mas-Colell \cite{hart1989potential} satisfies these properties. Given a cooperative game with a set $N$ of players and a cost function $C$, the potential function is the unique
function for which the sum of marginal contributions of every coalition equals its cost, i.e., for every $S\subseteq N$, $\sum_{i\in S}P_{C}(S)-P_{C}(S-\{i\})=C(S)$. 

The potential function has several interesting properties, e.g., the marginal contribution $P_C(i|S)$ coincides with the Shapley value of player $i$ in the coalition $S\cup \{i\}$. In addition, it is also known that the worst-case welfare loss of the Shapley mechanism for simple cost sharing is given by $P_{C}(OPT)$ \cite{moulin2001strategyproof}.

We generalize and adapt the potential function to our needs: the potential function as defined in \cite{hart1989potential} considers cooperative games, i.e., the cost function defined on subsets of $N$. Our generalization considers allocations. Specifically, we define the marginal contribution of player $i$ to the allocation $(S_{1},...,S_{n})$ by $P_{C}(S_{1},...,S_{i-1},S_{i},S_{i+1},...,S_{n})-P_{C}(S_{1},...,S_{i-1},\emptyset,S_{i+1},...,S_{n})$. 

We set the function $H$ as in (\ref{eq:VCG-Intro}) to be our generalization of the potential function, and name the new mechanism, a VCG mechanism using the potential function, the \emph{Potential Mechanism. }Notice that this gives a dominant strategy cost-recovering mechanism for \emph{every} cost function $C$. We are also able to prove efficiency guarantees if the cost function $C$ is subadditive. Combined together, we get the following general result:

\vspace{0.1in} \noindent \textbf{Theorem: }Let $C$ be a subadditive cost function. Then, the Potential mechanism always recovers the cost and provides an approximation ratio of $2\mathcal H_n$ to the social cost. If $C$ is submodular (or even XOS) the approximation ratio improves to $\mathcal H_n$.

\vspace{0.1in} \noindent Again, the approximation ratio is essentially tight due to the impossibility of \cite{dobzinski2008shapley}. For simple cost sharing, there is a dominant-strategy cost recovering mechanism that provides an approximation ratio of $\mathcal H_n$ for any cost function, by running the Shapley value mechanism ``on top'' of the VCG mechanism \cite{sundararajan2009trade}. Other works that focus on cost recovering mechanisms are \cite{georgiou2013black, fu2013cost}. In particular, \cite{georgiou2013black} handles some multi-parameter domains. However, their results are obtained via a reduction to the single parameter setting, e.g., by assuming additive valuations or considering ``all or nothing'' solutions. This approach leads to poor approximation ratios -- linear in the number of goods.

\subsection*{The Main Result}

The main technical effort of this work is in identifying three settings in which the potential mechanism is not only efficient and cost recovering, but also budget balanced.

\vspace{0.1in} \noindent \textbf{Theorem: }Let $C$ be a submodular cost function. The Potential Mechanism is $\mathcal{H}_{n}$-budget-balanced and provides an approximation ratio of $\mathcal{H}_{n}$ to the social cost in each of the following settings:
\begin{itemize}
\item Supermodular valuations.
\item General symmetric valuations and player-wise symmetric cost function\footnote{A valuation is symmetric if $v_i(S)=v_i(T)$ whenever $|S|=|T|$. A cost function is player-symmetric if $C(\vec S)=C(\vec T)$ whenever $|S_i|=|T_i|$ for all $i$.}.
\item Two players ($n=2$) with general valuations.
\end{itemize}

\vspace{0.1in} \noindent Some results were already known for some special cases of the symmetric setting, but even then submodular valuations are required: Mehta et al. \cite{mehta2007beyond} consider \emph{non-metric fault-tolerant uncapacitated facility location}. Bleischwitz and Schoppmann \cite{bleischwitz2008group} study \emph{metric fault tolerant uncapacitated facility location}. The mechanisms of Mehta et al. and of Bleischwitz and Schoppmann are actually generalizations of Moulin mechanisms \cite{moulin1999incremental,moulin2001strategyproof} and rely in their core on pushing single parameter cost sharing techniques to their limit. In contrast, in the symmetric setting our mechanism does not make any assumption on the valuations, except monotonicity. 


\subsection*{Limitations and Impossibilities}

Our results are tight in several aspects. First, even for excludable public good there is an instance in which the approximation ratio of the potential mechanism is $\Theta(\log n)$ and the mechanism is $\Theta(\log n)$ budget balanced. Actually, this is true for every VCG-based mechanism that uses a symmetric function\footnote{Note that the potential mechanism is symmetric when $C$ is symmetric, and in particular is symmetric for excludable public good.} (i.e., for $\vec S,\vec T$ with $|\vec S|=|\vec T|$, $H(\vec S)=H(\vec T)$): 

\vspace{0.1in} \noindent \textbf{Theorem: }Let $A$ be a symmetric VCG-based mechanism for the excludable public good problem (i.e., a mechanism that maximizes $\sum_{i}v_{i}(S_{i})-H(\vec{S})$ for some symmetric $H$) which always covers the incurred cost. Suppose that $A$ provides an approximation ratio of $n^{1-\epsilon}$ to the social cost, for some $\epsilon>0$. Then, there is an instance in which both the approximation ratio is $\Omega(\log n)$ and the sum of payments collected by the mechanism is $\Omega(\log n)$ times the incurred cost.

\vspace{0.1in} \noindent Next we show that one cannot forgo the assumption of a restricted cost function since there is a (non-subadditive) cost function for which the approximation ratio of the potential mechanism is $\Omega(n)$. Moreover, the potential mechanism is $\Omega(n)$-budget-balanced for a submodular cost function and unit demand valuations (a simple case of submodular valuations). A full mapping of the families for which the potential mechanism provides good approximation ratios is an open problem.

We also show that while one have many possible choices for the function $H$ in (\ref{eq:VCG-Intro}), the choice to use the potential function is not arbitrary: consider a cost-recovering VCG based mechanism $\sum_{i}v_{i}(S_{i})-H(\vec{S})$ where $H$ is normalized. Then, $H(\vec{S})\ge P_C(\vec{S})$ for every $\vec{S}$.

\subsection*{Open Questions}

Our work leaves many exciting questions open. We discuss some of them now.

\paragraph{Computational Issues.} In this work we focused on proving the existence of mechanisms with good guarantees. An interesting question is to understand whether there are mechanisms with good guarantees that are computationally efficient. We do not know the answer in general, but we do note that in the case of supermodular valuations and a submodular cost function the potential mechanism is computationally efficient whenever the value of the potential function is easy to compute. To see this, we point out that one of our results states that the potential function of a submodular cost function is submodular as well. Therefore, maximizing $\sum_{i}v_{i}(S_{i})-P_C(\vec{S})$ reduces to the problem of maximizing a supermodular function, which is equivalent to minimizing submodular function. Minimizing a submodular function can be done in polynomial time \cite{schrijver2000combinatorial} assuming oracle access to the potential function. 

Now observe that the potential function is easy to evaluate given the Shapley values (since the marginals of the potential function correspond to Shapley values). We conclude that in this case the potential mechanism is easy to compute whenever computing Shapley values is easy.

\paragraph{The Performance of the Potential Mechanism.} We identified three settings in which the potential mechanism is $\mathcal{H}_{n}$-budget-balanced and provides an approximation ratio
of $\mathcal{H}_{n}$ to the social cost. When the cost function is subadditive, the social cost approximation is $2\mathcal H_n$, but it is not clear whether we get budget balance for some interesting classes of valuations. A first step in this direction will be to determine the overcharging of the potential mechanism for simple cost sharing when the cost function is subadditive.

\paragraph{Other Mechanisms for Combinatorial Cost Sharing.} The potential mechanism is a VCG-based mechanism. Are there non VCG-based mechanisms with good guarantees, especially for settings in which the potential mechanism fails to deliver, e.g., submodular cost and valuation functions?

\paragraph{The Power of GSP vs. Strategyproof Mechanisms.} We know very little about groupstrategyproof mechanisms for combinatorial cost sharing. We provide a mechanism that guarantees a poor approximation ratio of $\Omega(n)$ (see Section \ref{sec-gsp}). Are there better groupstrategyproof mechanisms?

\paragraph{Impossibilities for Combinatorial Cost Sharing.}
The potential mechanism provides an approximation ratio of $\mathcal H_n$ to the social cost and is $\mathcal H_n$-budget-balanced for supermodular valuations and submodular cost functions. Is there a strategyproof combinatorial cost-sharing mechanism for submodular cost function and supermodular valuation functions which is $\beta$-budget-balanced and provides approximation ratio of $\rho$ to the social cost, where $\beta,\rho<\mathcal{H}_{n}$?

In particular, the only unrestricted impossibility result that applies to combinatorial cost sharing is the $\Omega(\log n)$ bound to the social cost of \cite{dobzinski2008shapley} which is proved for the restricted case of excludable public good. It seems that completely new ideas are needed to prove the optimality of the potential mechanism in richer combinatorial settings.


\section{Preliminaries}



In combinatorial cost sharing we have a set $N$ of $n$ players and $M=\bigcup_{i=1}^{n}M_{i}$ a set of services, where for
every $i\neq j$: $M_{i}\cap M_{j}=\emptyset$. Every player $i\in N$
has a valuation function
$v_{i}:2^{M_{i}}\to\mathbb{R}^{+}$. There is a cost
function $C:2^{M_{1}}\times\ldots\times2^{M_{n}}\to\mathbb{R}^{+}$ that
specifies the cost of every possible allocation of services. The output
is $(\vec{S},\vec{p})$ where $S_{i}\subseteq M_{i}$ is the allocation
of services to player $i$, and $p_{i}\ge0$ is player $i$'s payment. 

We assume that the cost function and the valuations
are \emph{monotone }(i.e., $S\subseteq T$ implies $C(S)\le C(T)$
and $v_{i}(S)\le v_{i}(T)$) and \emph{normalized }(i.e., $C(\emptyset)=0$
and $v_{i}(\emptyset)=0$).

We note that since the $M_{i}$'s are disjoint then $2^{M}=2^{M_{1}}\times\ldots\times2^{M_{n}}$.
Thus, we sometimes refer to functions from $2^{M_{1}}\times\ldots\times2^{M_{n}}$
as functions from $2^{M}$, e.g., $C(S_{1},\ldots,S_{n})$ will be denoted
$C\left(\bigcup_{i=1}^{n}S_{i}\right)$, where $S_{i}\subseteq M_{i}$
for each $1\le i\le n$. In addition, For a vector
of pairwise disjoint sets $\vec{S}=(S_{1},\ldots,S_{n})$ and a set
$T\subseteq M_i$ we let $\vec{S}-T$ be the vector
$(S_{1},...,S_{i-1},S_i-T,S_{i+1},...,S_{n})$. We similarly define the usual set operations (e.g., union) between two vectors of disjoint sets $\vec S,\vec T$ as applying this operation on all pairs of the form $(S_i, T_i)$.

Note that simple cost sharing is a special case of combinatorial cost sharing: set $|M_{i}|=1$
for every player $i$. In this case $|M|=|N|$ and we may think of $S\subseteq N$ as a set of served players: $i\in S$ if $S_{i}\neq\emptyset$. An important special case of simple cost sharing is \emph{excludable public good}: $C(\vec{S})=1$ for every $\vec{S}\neq(\emptyset,...,\emptyset)$ and $C(\emptyset)=0$. 


A (direct) combinatorial cost sharing mechanism takes the valuations of the players and outputs an allocation of services and prices. In this paper we consider only mechanisms that satisfy \emph{individual
rationality} (IR), if the mechanism charges player $i$ an amount $p_i$, then $p_{i}\le v_{i}(S_{i})$. A second requirement is \emph{no positive transfers} (NPT), meaning that $p_i\geq 0$. We now discuss other desired properties of mechanisms.

\paragraph{Incentive Compatibility} The focus of most this paper is in dominant-strategy mechanisms: a combinatorial cost-sharing mechanism is \emph{strategyproof}
or \emph{truthful} (in dominant strategies), if for every player $i$ with valuation function
$v_{i}$ and for every other valuation function $v_{i}^{\prime}$
it holds that $v_{i}(S_{i})-p_{i}\ge v_{i}(S_{i}^{\prime})-p_{i}^{\prime}$
where $(\vec{S},\vec{p})$ and $(\vec{S}^{\prime},\vec{p}^{\prime})$
are the outputs of the mechanism for $\vec{v}$ and $(v_{i}^{\prime},\vec{v}_{-i})$,
respectively. Groupstrategyproofness will be defined when needed, in Section \ref{sec-gsp}.

\paragraph*{Budget Balance.} A cost-sharing mechanism is \emph{budget-balanced} if the payments
charged from the players cover exactly the incurred cost, i.e. if
the mechanism's output is $(\vec{S},\vec{p})$ then $\sum_{i\in N}p_{i}=C(\vec{S})$.
This notion can be relaxed: a cost-sharing mechanism is \emph{$\beta$-budget-balanced} (for $\beta\ge1$) if for every valuation profile and every cost function the mechanism
outputs a services allocation $\vec{S}$ and payments vector $\vec{p}$
such that $C(\vec{S})\le\sum_{i\in N}p_{i}\le\beta\cdot C(\vec{S})$.

\paragraph{Economic Efficiency.} Roughgarden and Sundararajan \cite{roughgarden2009quantifying} suggested to quantify the inefficiency of cost-sharing mechanisms by the notion
of \emph{social cost, }denoted\emph{ $\pi(S)$}, which is the incurred
cost of the mechanism plus the excluded values. We extend the definition
of the social cost to the combinatorial setting: 
\[
\pi(\vec{S})=C(\vec{S})+\sum_{i\in N}\left[v_{i}(M_{i})-v_{i}(S_{i})\right]
\]

A cost sharing mechanism is a $\rho$-approximation to
the social cost if for every valuation profile the social cost of
its outcome $\vec{S}$ is at most $\rho\ge1$ times that of the allocation $\overrightarrow{OPT}$ that minimizes the social cost: $\pi(\vec{S})\le\rho\cdot\pi(\overrightarrow{OPT})$.

\subsubsection*{Cost and Valuation Functions}

A function $C:2^{M}\to\mathbb{R}^{+}$ is submodular if for every
two sets $S,T\subseteq M$, $C(S)+C(T)\ge C(S\cup T)+C(S\cap T)$.
Equivalently, a function $C$ is submodular if for every two sets
$T\subseteq S\subseteq M$ and an element $x\in M- S$ it
holds that $C(T\cup\left\{ x\right\} )-C(T)\ge C(S\cup\left\{ x\right\} )-C(S)$.
A cost function $C$ is \emph{subadditive }if for every $S,T\subseteq M$,
$C(S)+C(T)\ge C(S\cup T)$. A cost function $C$ is \emph{player-wise symmetric
}if for every two allocations $\vec{S},\vec{T}\in2^{M_{1}}\times\cdots\times2^{M_{n}}$
with $|S_{i}|=|T_{i}|$ for every $i\in N$ in holds that $C(\vec{T})=C(\vec{S})$.

$v$ is \emph{supermodular} if for every two sets $S,T\subseteq M_{i}$, it holds that $v(S)+v(T)\le v(S\cup T)+v(S\cap T)$.
Equivalently, for every two sets $T\subseteq S$ and
an element $x\in M- S$ it holds that $v(T\cup\{x\})-v(T)\le v(S\cup\{x\})-v(S)$.
A function $v:2^{M}\to\mathbb{R}$ is \emph{symmetric }if
for every two sets $S,T\subseteq M_{i}$ with $|S|=|T|$ it holds
that $v(S)=v(T)$. $v$ is \emph{additive} if for every $S$ it holds that
$v(S)=\sum_{j\in S}v(\{j\})$. $v$ is XOS if there exists additive valuations $a_1,\ldots, a_t$ such that for every $S$, $v(S)=\max_ra_r(S)$. The class of submodular valuations is contained in the class of XOS valuations which is contained in the class of subadditive valuations, and these containments are strict.


\section{The Potential Function and its Properties}\label{sec:function}

The potential mechanism is an affine maximizer whose output allocation is defined by:
\[
\arg\max_{\vec{S}=(S_1,\ldots, S_n)}\sum_{i}v_{i}(S_{i})-H(\vec{S})
\]
where $H:2^{M_{1}}\times\cdots \times2^{M_{n}}\to\mathbb{R}$ is a publicly known function. In this section we describe the function $H$ we use and its properties. 

The function $H$ we use is an adaptation of the \emph{Potential
Function} of Hart and Mas-Colell \cite{hart1989potential}. Given a cooperative game with $N$ players and cost function $C$, Hart and Mas-Colell define a potential function to be a normalized function in which the sum of discrete gradients of any set of players $N'$ equals the cost $C(N')$. We naturally extend this definition: 

\begin{definition}[potential function for combinatorial cost sharing]
Given a normalized cost function $C:2^{M_{1}}\times\cdots\times2^{M_{n}}\to\mathbb{R}^{+}$ we define a potential function $P_{C}:2^{M_{1}}\times\cdots\times2^{M_{n}}\to\mathbb{R}^{+}$ to be any normalized function with the following property:
\begin{itemize}
\item For every allocation $\vec{S}=(S_{1},...,S_{n})$ such that $S_{i}\subseteq M_{i}$,
the sum of discrete gradients of the function is exactly the cost: $\sum_{i\in N}\left[P_{C}\left(\vec{S}\right)-P_{C}\left(\vec{S}-S_{i}\right)\right]=C\left(\vec{S}\right)$.
\end{itemize}
\end{definition}
Notice that for simple cost sharing the definitions of \cite{hart1989potential} and ours coincide. The proof of the next proposition essentially follows from \cite{hart1989potential}. We bring it in Appendix \ref{appendix:ExplicitProof} for completeness.
\begin{proposition} \label{prop:uniquePc}
For every normalized cost function $C$ there is a unique potential
function $P_{C}$. Moreover, it holds that $P_{C}(\vec{S})=\sum_{I\subseteq N}\frac{C\left(\cup_{i\in I}S_i\right)}{|I|\cdot{n \choose \left|I\right|}}$.
\end{proposition}
Following \cite{hart1989potential}, there are a couple of useful observations related to the potential function. First, by rearranging:
$$P_{C}(\vec{S})=\sum_{I\subseteq N}\frac{C\left(\cup_{i\in I}S_i\right)}{|I|\cdot{n \choose \left|I\right|}}=\sum_{l=1}^n\sum_{I\subseteq N,|I|=l}\frac{C\left(\cup_{i\in I}S_i\right)}{|I|\cdot{n \choose \left|I\right|}}$$
Let $D(\vec S,l)=\sum_{I\subseteq N,|I|=l}\frac{C\left(\cup_{i\in I}S_i\right)}{|I|\cdot{n \choose \left|I\right|}}$. Notice the $D(\vec S,l)$ is the expected density (cost divided by the number of players) of a set $I$ of players that is chosen uniformly at random among all sets of size $l$, when each player $i\in I$ is served $S_i$ and player $i\notin I$ is not served at all. Thus, $P_C(\vec S)$ equals the sum of expected densities of sets of size $1,\ldots, n$.

A second observation (again, following \cite{hart1989potential}) is that the marginal utility $P_C(\vec S)-P_C(\vec S-S_i)$ equals the Shapley value of player $i$ in the cooperative game with $N$ players and cost function $C'$, where the cost of serving a set $T$ is $C'(T)=C(\cup_{i\in T}S_i)$.

We now use these observations to prove several useful properties of our potential function.

\begin{claim}\label{prop:C-submodular->PSubmodular}
Let $C$ be a monotone cost function. The potential function $P_C$ is monotone as well. Furthermore, if $C$ is submodular then $P_C$ is submodular as well.
\end{claim}
\begin{proof}
Observe that $P_C(\cdot)$ is a sum of functions ${C\left(\cup_{i\in I}S_i\right)}$, each multiplied by a constant $\frac 1 {|I|\cdot{n \choose \left|I\right|}}$. Each of these functions is monotone because $C$ is monotone. Since a sum of monotone functions is monotone as well, we conclude that $P_C$ is monotone. 

The second part of the claim is almost identical, by replacing ``monotone'' with ``submodular'' in the previous paragraph.
\end{proof}

Next we prove that for every submodular cost function $C$ and $\vec S$, $\mathcal H_n\cdot C(\vec S)\geq P_{C}(\vec S)\geq C(\vec S)$. We also prove a similar bound for subadditive cost functions. 
\begin{proposition}\label{prop-bound-below}
\label{prop:C<P}Let $C$ be an XOS (in particular, submodular) cost function. Then, for every $\vec S$, $P_{C}(\vec S)\geq C(\vec S)$. If $C$ is subadditive then for every $\vec S$, $P_C(\vec S)\geq \frac {C(\vec S)} 2$.
\end{proposition}
\begin{proof}
We will need the following folklore claim:
\begin{claim}\label{claim-random-subset}
Let $C':2^N\rightarrow \mathbb R$ be a monotone function. Choose a set $T$ of size $l$ uniformly at random among all such sets. Then:
\begin{enumerate}
\item If $C'$ is XOS then $\frac n l\cdot  E[C'(T)]\geq  C'(N)$.
\item If $C'$ is subadditive then $\lceil \frac n l\rceil \cdot  E[C'(T)]\geq  C'(N)$.
\end{enumerate}
\end{claim}
\begin{proof}(of Claim \ref{claim-random-subset})
\begin{enumerate}
\item By the definition of an XOS valuation, there is some additive valuation $a$ such that $a(N)=C'(N)$. Moreover, for every $S$, $C'(S)\geq a(S)$. Consider choosing at random a set $T$ of size $l$. Each element $j$ is selected with probability $\frac l N$. Let $A_j$ be the random variable that gets the value $a(\{j\})$ if $j$ is selected and $0$ otherwise. Observe that:
$$
E[C'(T)]\geq E[a(T)]= \Sigma_{j\in N}A_j=\Sigma_{j\in N}\frac l n\cdot a(\{j\})=\frac l n \cdot C'(N)
$$
\item Let $t=\lfloor \frac n l\rfloor$. Choose $t$ disjoint sets $T_1,\ldots, T_t$ each of size $l$, uniformly at random. If $n$ does not divide $l$, construct an additional set $T_{t+1}$ from taking $N-\cup_jT_j$ in addition to $l-|N-\cup_jT_j|$ items that are chosen uniformly at random. Now choose a set $R$ by choosing one of the sets $T_1,\ldots, T_{\lceil \frac n l\rceil}$ uniformly at random. Notice that by subadditivity $\Sigma_{i=1}^{\lceil \frac n l\rceil}C'(T_i)\geq C'(N)$. Thus, $E[R]\geq \frac 1 {{\lceil \frac n l\rceil }}\cdot C'(N)$. The proposition follows since $R$ is a set of size $l$ which is chosen uniformly at random.
\end{enumerate}
\end{proof}

We can now prove Proposition \ref{prop-bound-below}. Fix $\vec S=(S_1,\ldots, S_n)$. Obtain a function $C':2^N\rightarrow \mathbb R$ by $C'(I)=C(\cup_{i\in I}S_i)$. Observe that if $C$ is subadditive (XOS) then $C'$ is subadditive (XOS). For every $l$, let $D'(l)=D(\vec S,l)$ denote the expected density of a set of size $l$ in $C'$. Notice that by Claim \ref{claim-random-subset} for XOS functions we have that $D'(l)\geq \frac l n \cdot C'(N)$. Now, if $C'$ is XOS, then:
$$
P_{C}(\vec S)=\Sigma_{l=1}^nD'(l)\geq \Sigma_{l=1}^n  \frac {\frac l n \cdot C'(N)} l=\Sigma_{l=1}^n\frac {C'(N)} n=C'(N)
$$
The proof is similar if $C'$ is subadditive:
$$
P_{C}(\vec S)=\Sigma_{l=1}^nD'(l)\geq \Sigma_{l=1}^n\frac {C'(N)} {\lceil \frac n l \rceil \cdot l}\geq \frac {C'(N)} 2
$$
\end{proof}

\begin{claim}
\label{prop:P<HnC}Let $C$ be a monotone cost function. Then, for every $\vec S$, $P_{C}(\vec S)\le\mathcal{H}_{n}\cdot C(\vec S)$. 
\end{claim}
\begin{proof}
$C$ is monotone, so for every $\vec S$ and $l$ we have the following bound on the average density of a random set of size $l$: $D(\vec S,l)\leq \frac {C(\vec S)} l$. Thus, $P_C(\vec S)=\Sigma_{l=1}^nD(\vec S,l)\leq \Sigma_{l=1}^nC(\vec S)\cdot \frac 1 l=C(\vec S)\cdot \mathcal H_n$.
\end{proof}

\section{The Main Result: The Potential Mechanism}\label{sec:The-Mechanism}

In this section we present our main result, the Potential Mechanism. Specifically, given a cost function $C:2^{M_{1}}\times\ldots\times2^{M_{n}}\to\mathbb{R}^{+}$
and its potential function $P_{C}$, define the \emph{Potential Mechanism} be the following affine maximizer: 
\begin{equation}
\overrightarrow{ALG}=\arg\max_{\vec{S}=(S_{1},\ldots,S_{n})}\sum_{i\in N}v_{i}(S_{i})-P_{C}(\vec{S})\label{eq:mechanismVCG}
\end{equation}
We charge player $i$ his VCG payment which is given by the formula:
\begin{equation*}
p_{i}=\left[\sum_{j\in N}v_{j}(ALG_{j}^{-i})-P_C(\overrightarrow{ALG^{-i}})\right]-\left[\sum_{j\in N\backslash\{i\}}v_{j}(ALG_{j})-P_{C}(\overrightarrow {ALG})\right]
\end{equation*}
where $\overrightarrow{ALG^{-i}}=(ALG_{1}^{-i},\ldots,ALG_{i-1}^{-i},\emptyset,ALG_{i+1}^{-i},\ldots,ALG_{n}^{-i})$ is the allocation that maximizes (\ref{eq:mechanismVCG}) with the additional constraint $ALG^{-i}_{i}=\emptyset$. Note that there might be several allocations that maximize (\ref{eq:mechanismVCG}), hence the mechanism is defined up to the implementation of a tie
breaking rule. 
\begin{theorem}
\label{thm:Potential-Mechanism}For every valuation profile and cost function $C$ the Potential Mechanism is strategyproof, individually rational, always covers the incurred cost and satisfies no positive transfers. Moreover:

\begin{enumerate}
\item When $C$ is subadditive the mechanism provides an approximation ratio of $2\mathcal H_n$ to the social cost. When $C$ is XOS (in particular, submodular) the mechanism provides an approximation ratio of $\mathcal H_n$ to the social cost.
\item When $C$ is submodular the mechanism is $\mathcal H_n$-budget-balanced in each of the following settings:
\begin{enumerate}
\item \label{enu:theorem-Potential-Mechanism:Con1}The valuation functions
are supermodular.
\item \label{enu:Theorem-Potential-Mechanism-Con2}The valuation functions
are symmetric and the cost function $C$ is player-wise
symmetric. 
\item \label{enu:Theorem-Potential-Mechanism-Set3}There are two players.
\end{enumerate}
\end{enumerate}
\end{theorem}
Notice that the potential mechanism exhibits a ``one size fits all'' property: the performance guarantees improve as the more restricted the functions are although the mechanism is defined identically for all cost and valuation functions.

We note that strategyproofness, individual rationality, and no positive transfers of the potential mechanism follow directly from the properties of the VCG mechanism. The rest of the section is devoted to proving that the mechanism always recovers the cost and is approximately efficient and budget balanced in certain settings and in Section \ref{appendix-limitations-potential} we discuss some of the limitations of the potential mechanism.

\subsection*{The Potential Mechanism: Cost Recovering and Efficiency}

\begin{lemma}
	\label{lem:PotMech-no-deficit}The potential mechanism always covers the incurred cost,
	i.e. for every outcome $\left(\overrightarrow{ALG},\vec{p}\right)$,
	it holds that $\sum p_{i}\ge C(\overrightarrow{ALG})$.
\end{lemma}
\begin{proof}
The payment of player $i$ is given by:
\[
p_{i}=\left[\sum_{j\in N}v_{j}(ALG_{j}^{-i})-P_{C}\left(\overrightarrow{ALG^{-i}}\right)\right]-\left[\sum_{j\in N\backslash\{i\}}v_{j}(ALG_{j})-P_{C}(\overrightarrow{ALG})\right]
\]
	
Notice that $\sum_{j\in N}v_{j}(ALG_{j}^{-i})-P_{C}\left(\overrightarrow{ALG^{-i}}\right)\ge\sum_{j\in N\backslash\{i\}}v_{j}(ALG_{j})-P_{C}(\overrightarrow{ALG}-ALG_{i})$ since $\overrightarrow{ALG}-ALG_{i}$ is one of the alternatives that are considered for $ALG^{-i}$, therefore 
	\begin{eqnarray*}
		p_{i} & \ge & \left[\sum_{j\in N\backslash\{i\}}v_{j}(ALG_{j})-P_{C}(\overrightarrow{ALG}-ALG_{i})\right]-\left[\sum_{j\in N\backslash\{i\}}v_{j}(ALG_{j})-P_{C}(\overrightarrow{ALG})\right]\\
		& = & P_{C}(\overrightarrow{ALG})-P_{C}(\overrightarrow{ALG}-ALG_{i})
	\end{eqnarray*}
	summing over all players, we get 
	\[
	\sum_{i\in N}p_{i}\ge\sum_{i\in N}P_{C}(\overrightarrow{ALG})-P_{C}(\overrightarrow{ALG}-ALG_{i})=C(\overrightarrow{ALG})
	\]
	where the last equality is by the definition of the potential function.
\end{proof}

Next we prove the efficiency guarantees of the mechanism. For completeness we prove guarantees for both the social cost and additive approximations, although the latter implies the former \cite{roughgarden2009quantifying}.

\begin{lemma}\label{lem:HnEFFApprox}
Let $C$ be a subadditive cost function. Then the approximation ratio of the Potential Mechanism to the social cost is $2\mathcal{H}_{n}$. If $C$ is XOS then the approximation ratio improves to $\mathcal H_n$. Similarly, $SW(\overrightarrow{OPT})-SW(\overrightarrow{ALG})\leq (\mathcal H_n-1)\cdot C(\overrightarrow{OPT})$ if $C$ is XOS and $SW(\overrightarrow{OPT})-SW(\overrightarrow{ALG})\leq(2\mathcal H_n-1)\cdot C(\overrightarrow{OPT})$ if $C$ is subadditive.
\end{lemma}
\begin{proof}
We first prove the social cost guarantee for submodular $C$:
\begin{eqnarray*}
\pi(\overrightarrow{ALG}) & = & \sum_{i\in N}\left[v_{i}(M_{i})-v_{i}(ALG_{i})\right]+C(\overrightarrow{ALG})\\
 & \le & \sum_{i\in N}\left[v_{i}(M_{i})-v_{i}(ALG_{i})\right]+P_{C}(\overrightarrow{ALG})\\
 & \le & \sum_{i\in N}\left[v_{i}(M_{i})-v_{i}(OPT_{i})\right]+P_{C}(\overrightarrow{OPT})
  \le  \mathcal{H}_{n}\cdot\pi(\overrightarrow{OPT})
\end{eqnarray*}

The first inequality is due to the first part of Proposition \ref{prop:C<P}: for every
XOS cost function $C$, $C(\overrightarrow{ALG})\le P_{C}(\overrightarrow{ALG})$.
The second inequality holds since $\overrightarrow{ALG}$ maximizes
(\ref{eq:mechanismVCG}), therefore $\sum_{i\in N}v_{i}(OPT_{i})-P_{C}(\overrightarrow{OPT})\le\sum_{i\in N}v_{i}(ALG_{i})-P_{C}(\overrightarrow{ALG})$.
The last inequality is due to Proposition \ref{prop:P<HnC} which
states that $P_{C}(\overrightarrow{OPT})\le\mathcal{H}_{n}\cdot C(\overrightarrow{OPT})$. The proof for subadditive $C$ is very similar except that we use the second part of Proposition \ref{prop:C<P}.

The proof for the additive approximation guarantee is similar (using the same propositions): 
\begin{eqnarray*}
	SW(\overrightarrow{OPT})-SW(\overrightarrow{ALG})
	& = & \sum_{i\in N}v_i(OPT_i) - C(\overrightarrow{OPT}) - 
	\left[\sum_{i\in N}v_i(ALG_i) - C(\overrightarrow{ALG}) \right] \\
	& \le & \sum_{i\in N}v_i(OPT_i) - C(\overrightarrow{OPT}) - 
	\left[\sum_{i\in N}v_i(ALG_i) - P_C(\overrightarrow{ALG}) \right] \\
	& \le & \sum_{i\in N}v_i(OPT_i) - C(\overrightarrow{OPT}) - 
	\left[\sum_{i\in N}v_i(OPT_i) - P_C(\overrightarrow{OPT}) \right] \\
	& = & P_C(\overrightarrow{OPT}) - C(\overrightarrow{OPT}) 
	 \le  (\mathcal{H}_n - 1)\cdot C(\overrightarrow{OPT})
\end{eqnarray*}
\end{proof}

Interestingly, up until now we have not used the property that the $M_i$'s are disjoint. That is, consider the usual setting of a combinatorial auction with $n$ players and a set of $M$ items. Let $C$ be a function that specifies the cost of every allocation. The same proofs imply that the potential mechanism always recovers the cost in this general setting. Furthermore, if $C$ is subadditive then the output is approximately efficient.

\subsection*{The Potential Mechanism: Budget Balance}

Most of the technical difficulty in the proof of Theorem \ref{thm:Potential-Mechanism} is showing that the mechanism is $\mathcal{H}_{n}$-budget-balanced. We have already shown that the potential mechanism always cover the cost (Lemma \ref{lem:PotMech-no-deficit}). We remain with showing that the sum of payments of the mechanism is bounded from above by $\mathcal{H}_{n}$ times the incurred cost.
 
\begin{lemma}
\label{lem:BB}In each one of the settings (\ref{enu:theorem-Potential-Mechanism:Con1}), (\ref{enu:Theorem-Potential-Mechanism-Con2}) and (\ref{enu:Theorem-Potential-Mechanism-Set3}) of theorem
\ref{thm:Potential-Mechanism} the following holds: for every outcome
of the mechanism $(\overrightarrow{ALG},\vec{p})$ we have that $\sum_{i\in N}p_{i}\le\mathcal{H}_{k}\cdot C(\overrightarrow{ALG})$.
\end{lemma}

Recall that the potential mechanism is defined up to a tie-breaking rule. For supermodular valuations and $n=2$ the proofs below work directly for any tie breaking rule. For the symmetric setting, we start with an arbitrary tie-breaking rule and modify it. We prove budget balance for the modified tie-breaking rule and see that this implies budget balance for the original tie breaking rule. The modified rule is defined as follows: for every player $i$ define an arbitrary order on $M_{i}$ and denote its elements as $M_{i}=\left\{ 1^{i},\ldots,|M_{i}|^{i}\right\} $. If player $i$ gets $s$ items in the original rule then allocate player $i$ items $\{1^i,\ldots, s^i\}$. Notice that due to symmetry all costs and values remain the same, so this new allocation still maximizes the affine maximizer. Finally, observe that since that player $i$ is always indifferent between the bundles he receives in the original and modified rules, we may assume that the price of these bundles is identical as well, so his payment remains the same. Hence, budget balance for the modified rule implies budget balance for the original rule as well.


In order to prove lemma \ref{lem:BB} we use the following claim:
\begin{claim}
\label{claim:helps-lemma-BB}
In each one of settings (\ref{enu:theorem-Potential-Mechanism:Con1}),  (\ref{enu:Theorem-Potential-Mechanism-Con2}) and (\ref{enu:Theorem-Potential-Mechanism-Set3}) of theorem \ref{thm:Potential-Mechanism}:
\begin{equation}\label{eq:claim-helps-lemma-BB-equation}
\sum_{i=1}^{n}\left[\sum_{j\in N}v_{j}(ALG_{j}^{-i})-P_{C}(\overrightarrow{ALG^{-i}})\right]\le(n-1)\cdot\left[\sum_{j=1}^{n}v_{j}(ALG_{j})-P_{C}(\overrightarrow{ALG})\right]
\end{equation}
\end{claim}
Before providing the proof of claim \ref{claim:helps-lemma-BB} in subsection \ref{sub:Proof-of-Claim}, we show how it implies lemma
\ref{lem:BB}.
\begin{proof}(of lemma \ref{lem:BB})
Consider the sum of the payments:
\begin{equation*}
\sum_{i=1}^{n}p_{i}=\sum_{i=1}^{n}\left[\left[\sum_{j\in N}v_{j}(ALG_{j}^{-i})-P_{C}(\overrightarrow{ALG^{-i}})\right]-\left[\sum_{j\in N\backslash\{i\}}v_{j}(ALG_{j})-P_{C}(\overrightarrow{ALG})\right]\right]\label{eq:Lemma-BB-1}
\end{equation*}
by claim \ref{claim:helps-lemma-BB} the RHS is bounded from above
by:
\[
(n-1)\cdot\left[\sum_{j=1}^{n}v_{j}(ALG_{j})-P_{C}(\overrightarrow{ALG})\right]-\sum_{i=1}^{n}\left[\sum_{j\in N\backslash\{i\}}v_{j}(ALG_{j})-P_{C}(\overrightarrow{ALG})\right]
\]
re-arranging, we observe that $\sum_{i=1}^{n}\sum_{j\in N\backslash\{i\}}v_{j}(ALG_{j})=(n-1)\cdot\sum_{j=1}^{n}v_{j}(ALG_{j})$.
Therefore the above equals to:

\[
(n-1)\cdot\left[\sum_{j=1}^{n}v_{j}(ALG_{j})-P_{C}(\overrightarrow{ALG})\right]-\left(n-1\right)\cdot\left[\sum_{j=1}^{n}v_{j}(ALG_{j})\right]+n\cdot P_{C}(\overrightarrow{ALG})=P_{C}(\overrightarrow{ALG})
\]
By proposition \ref{prop:P<HnC} we get that
$P_{C}(\overrightarrow{ALG})\le\mathcal{H}_{n}\cdot C(\overrightarrow{ALG})$.
Hence, $\sum_{i=1}^{n}p_{i}\le\mathcal{H}_{n}\cdot C(\overrightarrow{ALG})$.

\end{proof}

\subsection{\label{sub:Proof-of-Claim}Proof of Claim \ref{claim:helps-lemma-BB}}

The proof consists of several steps: Claim \ref{claim:bound potential}
shows a lower bound on $\sum_{i=1}^{n}P_{C}(\overrightarrow{ALG^{-i}})$, Claim \ref{claim:bound valuations-supermodular} gives an upper bound on $\sum_{i=1}^{n}\sum_{j\in N}v_{j}(ALG_{j}^{-i})$ for supermodular valuations and $n=2$, and claim \ref{claim:bound-v-multi-unit} gives a similar bound for the symmetric setting. The following technical claim will be used:
\begin{claim}
\label{claim:supermodularBoundbyN-1}Let $H:2^{M}\to\mathbb{R}$
be a normalized function, and consider $\left\{ T^{i}\right\} _{i=1}^{n}$
such that for every $i$, $T^{i}\in2^{M}$ and  $\bigcap_{i=1}^{n}T^{i}=\emptyset$.
Suppose that either (1) $H$ is supermodular or (2) for every $1\le i,j\le n$,
$T^{i}\cup T^{j}\in\left\{ T^{i},T^{j}\right\} $. Then it holds that
$\sum_{i=1}^{n}H(T^{i})\le\sum_{i=1}^{n-1}H\left(T^{i+1}\cup\left(\bigcap_{\ell=1}^{i}T^{\ell}\right)\right)$.
\end{claim}
\begin{proof}
If $H$ is supermodular, then beginning with the first two summands
of the LHS of the claim, by supermodularity of \emph{H:} 
\[
H(T^{1})+H(T^{2})\le H(T^{1}\cup T^{2})+H(T^{1}\cap T^{2})
\]
shifting our focus to $H(T^{1}\cap T^{2})$ and the third summand
$H(T^{3})$ yields
\[
H(T^{1}\cap T^{2})+H(T^{3})\le H\left(\left(T^{1}\cap T^{2}\right)\cup T^{3}\right)+H\left(T^{1}\cap T^{2}\cap T^{3}\right)
\]
iterating the above for $n$ times shows 
\[
H\left(T^{1}\cap\ldots\cap T^{n-1}\right)+H(T^{n})\le H\left(\left(\bigcap_{i=1}^{n-1}T^{i}\right)\cup T^{n}\right)+H\left(\bigcap_{i=1}^{n}T^{i}\right)
\]
by combining everything we get 
\[
\sum_{i=1}^{n}H(T^{i})\le\sum_{i=1}^{n-1}H\left(T^{i+1}\cup\left(\bigcap_{\ell=1}^{i}T^{\ell}\right)\right)+H\left(\bigcap_{i=1}^{n}T^{i}\right)
\]
since $\bigcap_{i=1}^{n}T^{i}=\emptyset$ and $H\left(\emptyset\right)=0$,
the claim follows. 

If $T^{i}\cup T^{j}\in\{T^{i},T^{j}\}$ for every $1\le i,j\le n$,
then without loss of generality when $T^{i}\cup T^{j}=T^{j}$ we have that $T^{i}\cap T^{j}=T^{i}$. Thus, a similar proof can be applied
on $\left\{ T^{i}\right\} _{i\in N}$ with strict equality.
\end{proof}

The next claim bounds the sum of the $P_{C}$'s as they appear in Claim \ref{claim:helps-lemma-BB}. 
\begin{claim}
\label{claim:bound potential}
Suppose that the cost function $C$ is submodular and normalized, then
\[
-\sum_{i=1}^{n}P_{C}(\overrightarrow{ALG^{-i}})\le-\sum_{i=1}^{n-1}P_{C}\left(\overrightarrow{ALG^{-[i+1]}}\bigcup\left(\cap_{\ell=1}^{i}\overrightarrow{ALG^{-\ell}}\right)\right)
\]
\end{claim}
\begin{proof}
Since $C$ is submodular, by Proposition \ref{prop:C-submodular->PSubmodular}
it holds that $P_{C}$ is submodular as well. Hence, $-P_{C}$ is
supermodular and we can apply Claim \ref{claim:supermodularBoundbyN-1}
with $-P_{C}$ and $\left\{\overrightarrow{ ALG^{-i}}\right\} _{i=1}^{n}$ (observe that for every $i$, $ALG_{i}^{-i}=\emptyset$ hence $\bigcap_{i=1}^{n}\overrightarrow{ALG^{-i}}=\emptyset$). 
\end{proof}

The next two claims bound the sum of valuations $\sum_{i=1}^{n}\sum_{j=1}^{n}v_{j}(ALG_{j}^{-i})$.
Each of the claims refers to different settings of Theorem \ref{thm:Potential-Mechanism}.
\begin{claim}\label{claim:bound valuations-supermodular}
Suppose that each $v_{i}$ is supermodular or that $n=2$. Then
\[
\sum_{i=1}^{n}\sum_{j=1}^{n}v_{j}(ALG_{j}^{-i})\le\sum_{i=1}^{n-1}\sum_{j=1}^{n}v_{j}\left(ALG_{j}^{-[i+1]}\bigcup\left(\cap_{\ell=1}^{i}ALG_{j}^{-\ell}\right)\right)
\]
\end{claim}
\begin{proof}
Define $V:2^{M_{1}}\times\cdots\times2^{M_{n}}\to\mathbb{R}$ by $V(\vec{T})=\sum_{j=1}^{n}v_{j}(T_{j})$. It is easy to see that $V$ is normalized. Also, since $ALG_{i}^{-i}=\emptyset$, it holds that $\bigcap_{i=1}^{n}\overrightarrow{ALG^{-i}}=\emptyset$. Observe that if the $v_i$'s are supermodular then $V$ is supermodular as well. Applying claim \ref{claim:supermodularBoundbyN-1} with $V$ and $\left\{\overrightarrow{ALG^{-i}}\right\} _{i=1}^{n}$:
\[
\sum_{i=1}^{n}V(\overrightarrow{ALG^{-i}})\le\sum_{i=1}^{n-1}V\left(\overrightarrow{ALG^{-[i+1]}}\bigcup\left(\cap_{\ell=1}^{i}\overrightarrow{ALG^{-\ell}}\right)\right)=\sum_{i=1}^{n-1}\sum_{j=1}^{n}v_{j}\left(\overrightarrow{ALG_{j}^{-[i+1]}}\bigcup\left(\cap_{\ell=1}^{i}\overrightarrow{ALG_{j}^{-\ell}}\right)\right)
\]
For $n=2$ we have that:
$$
\sum_{i=1}^{2}V(\overrightarrow{ALG^{-i}})= v_1(ALG^{-2}_1)+v_2(ALG^{-1}_2)=\sum_{j=1}^{2}v_{j}\left(ALG_{j}^{-2}\bigcup ALG_{j}^{-1}\right)
$$
\end{proof}

\begin{claim}
\label{claim:bound-v-multi-unit}Suppose that each $v_{i}$ is symmetric and $C$ is a player-wise symmetric cost function. Then: 
\[
\sum_{i=1}^{n}\sum_{j=1}^{n}v_{j}(ALG_{j}^{-i})\le\sum_{i=1}^{n-1}\sum_{j=1}^{n}v_{j}\left(ALG_{j}^{-[i+1]}\bigcup\left(\cap_{\ell=1}^{i}ALG_{j}^{-\ell}\right)\right)
\]
\end{claim}
\begin{proof}
%
We prove the claim assuming the tie breaking rule described earlier. Consider the LHS of the claim and change the order of summation: 
\[
\sum_{i=1}^{n}\sum_{j=1}^{n}v_{j}(ALG_{j}^{-i})=\sum_{j=1}^{n}\sum_{i=1}^{n}v_{j}(ALG_{j}^{-i})
\]
Fix player $j$ and consider $\sum_{i=1}^{n}v_{j}(ALG_{j}^{-i})$.
By the implementation, for every $1\le i\le n,\,\, ALG_{j}^{-i}=\{1^{j},\ldots,|ALG_{j}^{-i}|^{j}\}$.
Therefore, for every $1\le i,k\le n$ it holds that $ALG_{j}^{-i}\cup ALG_{j}^{-k}\in\left\{ ALG_{j}^{-i},ALG_{j}^{-k}\right\} $.
Moreover, since $ALG_{j}^{-j}=\emptyset,$ $\bigcap_{i=1}^{n}ALG_{j}^{-i}=\emptyset$.
Applying claim \ref{claim:supermodularBoundbyN-1} on $v_{j}$ and
$\left\{ ALG_{j}^{-i}\right\} _{i=1}^{n}$:
\[
\sum_{i=1}^{n}v_{j}(ALG_{j}^{-i})\le\sum_{i=1}^{n-1}v_{j}\left(ALG_{j}^{-[i+1]}\bigcup\left(\cap_{\ell=1}^{i}ALG_{j}^{-\ell}\right)\right)
\]
Summing over all players, we get 
\[
\sum_{j=1}^{n}\sum_{i=1}^{n}v_{j}(ALG_{j}^{-i})\le\sum_{j=1}^{n}\sum_{i=1}^{n-1}v_{j}\left(ALG_{j}^{-[i+1]}\bigcup\left(\cap_{\ell=1}^{i}ALG_{j}^{-\ell}\right)\right)
\]
Changing the order of summation, the RHS equals to 
\[
\sum_{i=1}^{n-1}\sum_{j=1}^{n}v_{j}\left(ALG_{j}^{-[i+1]}\bigcup\left(\cap_{\ell=1}^{i}ALG_{j}^{-\ell}\right)\right)
\]
\end{proof}

We can now finish the proof of claim \ref{claim:helps-lemma-BB}.

\begin{proof}(of claim \ref{claim:helps-lemma-BB}) By combining claims \ref{claim:bound potential}, \ref{claim:bound valuations-supermodular} and
\ref{claim:bound-v-multi-unit}:
\begin{eqnarray*}
 & \sum_{i=1}^{n}\left[\sum_{j=1}^{n}v_{j}(ALG_{j}^{-i})-P_{C}(\overrightarrow{ALG^{-i})}\right]\\
 & \le\sum_{i=1}^{n-1}\left[\sum_{j=1}^{n}v_{j}\left(ALG_{j}^{-[i+1]}\bigcup\left(\cap_{\ell=1}^{i}ALG_{j}^{-\ell}\right)\right)-P_{C}\left(\overrightarrow{ALG^{-[i+1]}}\bigcup\left(\cap_{\ell=1}^{i}\overrightarrow{ALG^{-\ell}}\right)\right)\right]
\end{eqnarray*}
We define $n-1$ allocations: for every $1\le i\le n-1$ we let $\overrightarrow{S^{i}}=\overrightarrow{ALG^{-[i+1]}}\cup\left(\bigcap_{\ell=1}^{i}\overrightarrow{ALG^{-\ell}}\right)$.
$\overrightarrow{ALG}$ maximizes (\ref{eq:mechanismVCG}) so $\sum_{j=1}^{n}v_{j}(S_{j}^{i})-P_{C}(\overrightarrow{S^{i}})\le\sum_{j=1}^{n}v_{j}(ALG_{j})-P_{C}(\overrightarrow{ALG})$
for every $\overrightarrow{S^{i}}$. Hence:
\[
\sum_{i=1}^{n}\left[\sum_{j\neq i}v_{j}(ALG_{j}^{-i})-P_{C}(\overrightarrow{ALG^{-i}})\right]\le(n-1)\cdot\left(\sum_{j=1}^{n}v_{j}(ALG_{j})-P_{C}(\overrightarrow{ALG})\right)
\]
\end{proof}

%
%

\section{Limitations of VCG-Based Mechanisms}

In this section we discuss some limitations of VCG-based mechanisms. We first ask whether our analysis of the potential mechanism is tight and prove that this is indeed the case. Our result is in fact more general: no symmetric VCG-based mechanism can do much better. We then show that the potential function is the ``minimal'' choice for $H(\cdot)$ in the affine maximizer in a concrete technical sense. 

A word is in place regarding the definition of an affine maximizer. In general an affine maximizer does not have to maximize over all allocations. This can be implemented by setting $H(\vec S)=\infty$ for an allocation $\vec S$ that is never selected by the affine maximizer. However, it is more convenient to work with a function $H$ that is monotone. This can be assumed without loss of generality: suppose that there is $\vec S, j\in M$ such that $H(\vec S)>H(\vec S+\{j\})$. Note that by the monotonicity of the valuations this implies that the affine maximizer never outputs $\vec S$. To make $H$ monotone, we can set $H(\vec S)=H(\vec S+\{j\})$ -- we still have that the affine maximizer always gives higher value to $\vec S+\{j\}$ and in case of equality we can assume that $\vec S$ is not selected by defining an appropriate tie breaking rule.

\subsubsection*{A Lower Bound for Symmetric VCG-Based Mechanisms}

We show that every VCG-based mechanism that always covers the cost and uses a symmetric function does not provide better guarantees than the potential mechanism for the excludable public good problem. Note that the potential mechanism is symmetric in this setting since the potential function is symmetric if $C$ is symmetric.

\begin{theorem}
\label{thm:VCG-Based-Lower-Bound}
Let $\mathcal{A}$ be a VCG-based mechanism for the excludable public good problem. I.e. $\mathcal{A}$ is a mechanism that outputs
\begin{equation}
\label{eq:vcg-based-peg}
ALG=\arg\max_{S\subseteq N}\sum_{i\in S}v_{i} - H(S)
\end{equation}
where $H:2^{N}\to\mathbb{R}^{+}\cup \{\infty\}$ is normalized, monotone and symmetric function
(i.e. for every two allocations $S,T\subseteq N$ with $\left|S\right|=\left|T\right|$
it holds that $H(S)=H(T)$). Suppose that $\mathcal{A}$ always covers the cost and provides an approximation ratio of $\rho<\frac{n^{(1-\delta)/2}}{4}$ to the social cost, for some constant $0<\delta<1$. Then, there is an instance in which the good is constructed at cost $1$, the sum of payments is at least $\left(\frac{1-\delta}{2}\cdot\log{n}-1\right)$ and the approximation ratio to the social cost is no better than $\left(\delta\cdot\log{n}-1\right)$.
\end{theorem}
In order to prove the theorem, we need the following claim. 
\begin{claim}\label{claim:VCG-Based-HS-HS-i>1/s}
Let $\mathcal{A}$ be a VCG-based mechanism for the excludable public good problem where $H:2^{N}\to\mathbb{R}^{+}$
is normalized, monotone and symmetric. Suppose that $\mathcal{A}$
always covers the cost and provides a finite approximation ratio  $\rho$ to the social cost. Then, for every $S\subseteq N$ and $i\in S$ it holds that $H(S)-H(S-\{i\})\ge\frac{1}{|S|}$.
\end{claim}
\begin{proof}
	Fix a set $S\subseteq N$ and consider the valuation profile
	in which $v_{i}>\rho$ if $i\in S$ and $v_{i}=0$ otherwise.
	Let $ALG\subseteq N$ be the output of the mechanism. Note that $S \subseteq ALG$,
	otherwise, there exists $i\in S\backslash ALG$ with value $v_i>\rho$, so the mechanism outputs an allocation with social cost $\pi(ALG)\ge v_i >\rho$, whereas the optimal allocation is $OPT=N$ with social cost $\pi(OPT)=\pi(N)=C(N)=1$, a contradiction to the approximation guarantee of $\mathcal{A}$. Therefore $\sum_{i\in S}v_{i}=\sum_{i\in ALG}v_{i}$ and by monotonicity of $H$, $H(S)\le H(ALG)$. Furthermore,  $H(S)\ge H(ALG)$ since $ALG$ maximizes (\ref{eq:vcg-based-peg}), thus $H(S)=H(ALG)$. By the same arguments, for every $i\in S$ it holds that $S-\{i\}\subseteq ALG^{-i}$ and $H(S-\{i\})=H(ALG^{-i})$. Therefore:
	\begin{eqnarray*}
		\sum_{j \in N} p_{j} &  =  & \sum_{j \in S} \left [ \left[ \sum_{i \in ALG^{-j}} v_i - H(ALG^{-j}) \right] - \left[  \sum_{i\in ALG-\{j\}} v_i  - H(ALG) \right] \right] \\
		& =  & \sum_{j \in S} \left[ \left[ \sum_{i \in S-\{j\}} v_i - H(S-\{j\}) \right] - \left[  \sum_{i\in S-\{j\}} v_i  - H(S) \right] \right] \\
		& = & \sum_{j \in S} \left[ H(S) - H(S-\{j\}) \right]
	\end{eqnarray*}
	Since $H$ is symmetric, for every $i,j\in S$ it holds that $H(S-\{i\})=H(S-\{j\})$.
	Fix a player $i\in S$. Since the mechanism covers the cost it holds
	that 
	\[
	\sum_{j\in S}p_{j}=|S|\cdot\left[H(S)-H(S-\{i\})\right]\ge C(ALG) = 1
	\]
	Therefore, $H(S)-H(S-\{i\})\ge\frac{1}{|S|}$.
\end{proof}

We now turn to proving Theorem \ref{thm:VCG-Based-Lower-Bound}. 

\begin{proof}(of theorem \ref{thm:VCG-Based-Lower-Bound}) Note that since the mechanism provides a $\rho$ approximation, $H(N)<\infty$. Otherwise, for every $i$, set $v_i=2\rho$. The optimal social cost is $1$, serve all players, but at least one player is excluded so $\pi(\overrightarrow{ALG})\geq 2\rho$, a contradiction.
%
Fix a set $S \subseteq N$ of size $n^{1-\delta}$. We now divide the analysis, depending
on whether there is some set $T$, with $|T|<n^{(1-\delta)/2}$
such that $\frac{H(T)}{|T|}\le\frac{H(S)}{|S|}$. 

Assume that there is a set $T$ with $|T|<n^{(1-\delta)/2}$ such that $\frac{H(T)}{|T|}\le\frac{H(S)}{|S|}$.
In this case we get that $H(S)\geq n^{(1-\delta)/2}\cdot H(T)$. By the monotonicity of $H$ and claim \ref{claim:VCG-Based-HS-HS-i>1/s} it holds that $H(T)\ge H(\{1\})\ge 1$, therefore $H(S)\ge n^{(1-\delta)/2}$.

Choose an arbitrary set $S^{\prime}\subseteq N$ with $S^{\prime}=2\cdot n^{1-\delta}<n$. Consider the following valuation profile: set $v_{i}=\frac{1}{4 \cdot n^{(1-\delta)/2}}$ if $i\in S^{\prime}$, and $v_{i}=0$ otherwise. Observe that in this profile the mechanism does not output a set of size more than $n^{1-\delta}$: consider some set $S^{\prime\prime}\subseteq S^{\prime}$ with $|S^{\prime\prime}| \ge n^{1-\delta}$. By monotonicity of $H$ it  holds that $H(S^{\prime\prime}) \ge H(S)\ge n^{(1-\delta)/2}$. Moreover it holds that $\sum_{i\in S^{\prime\prime}}v_{i} \le 2\cdot n^{1-\delta}\cdot \frac{1}{4\cdot n^{(1-\delta)/2}} \le \frac{n^{(1-\delta)/2}}{2} $ ,
hence $\sum_{i\in S^{\prime\prime}}v_{i}-H(S^{\prime\prime})<0$. 

Therefore, the mechanism outputs a set of size at most
$n^{1-\delta}$, meaning that at least $n^{1-\delta}$ players among $S^{\prime}$ with value $\frac{1}{4\cdot n^{(1-\delta)/2}}$
are not served. We get that the optimal optimal allocation is $\pi(OPT)=\pi(N)=1$, but the social cost of the output of the mechanism is at least $\frac{n^{(1-\delta)/2}}{4}$. Hence the approximation ratio of the mechanism is no better than $\frac{n^{(1-\delta)/2}}{4}$.

Therefore, from now on assume that for every $T$ with $|T|<n^{(1-\delta)/2}$ it holds that $\frac{H(T)}{|T|}>\frac{H(S)}{|S|}$. We construct a valuation profile in which the mechanism does not serve players whose sum of values is at least $\delta\cdot\log{n} - 1$, hence the approximation ratio of the social cost is no better than $\delta\cdot\log{n} -1 $. In addition, in the same profile the total sum of payments that the mechanism collects is at least $\frac{1-\delta}{2}\cdot\log{n} - 1$.

Let $W$ be the minimal set that minimizes the density $\frac{H(W)}{|W|}$ among
all sets of size at most $n^{1-\delta}$, i.e., for every $W^\prime \subset W$ it holds that $\frac{H(W)}{|W|} < \frac{H(W^\prime)}{|W^\prime|}$. Notice that $|W|\geq n^{(1-\delta)/2},$
since the density of \emph{S}, which is of size $|S|=n^{1-\delta}$, is smaller than the density of any set of size at most $n^{(1-\delta)/2}$.

Consider the following instance: for every $i\in W$ set $v_{i}=\frac{H(W)}{|W|} + \varepsilon$ where  $\varepsilon>0$ is such that $\varepsilon < \min\left\{\frac{1}{n^3},\min_{W^\prime \subset W} \left\{ \frac{H(W^\prime)}{|W^\prime|} - \frac{H(W)}{|W|} \right\}\right\}\ $. We get that $\sum_{i\in W}v_i > H(W)$, but for every $W^\prime \subset W$, it holds that 
$\sum_{i\in W^\prime} v_i = |W^\prime| \cdot \left[ \frac{H(W)}{|W|}+\varepsilon\right]<|W^\prime|\cdot\frac{H(W)}{|W|} + H(W^\prime)-|W^\prime|\cdot \frac{H(W)}{|W|}=H(W^\prime)$. In addition, consider a set $L\subseteq N\backslash W$, with $n-n^{1-\delta}$ players.
Without loss of generality, $L=\left\{ n^{1-\delta}+1,n^{1-\delta}+2,...,n\right\} $. For every
$i\in L$ set $v_{i}=\frac{1}{i}-\varepsilon$.

We observe that the mechanism outputs $ALG=W$ and that for every $i\in W$, $ALG^{-i}=\emptyset$: by claim \ref{claim:VCG-Based-HS-HS-i>1/s}, for any two disjoint
sets $W^\prime,L^\prime\subseteq N$, $H(W^\prime\cup L^\prime)-H(W^\prime)\ge\sum_{i=|W^\prime|+1}^{|W^\prime|+|L^\prime|}\frac{1}{i}$,
therefore for every subset $L^\prime\subseteq L$ and for every subset $W^\prime\subseteq W$
(which is of size at most $n^{1-\delta}$) the mechanism prefers the allocation
$W^\prime$ over $W^\prime\cup L^\prime$: 
\begin{eqnarray*}
		\sum_{i\in W^\prime\cup L^\prime}v_{i}-H(W^\prime\cup L^\prime) 
& \le & \sum_{i\in  W^\prime}v_{i}+\sum_{i\in L^\prime}v_{i}-\sum_{i=| W^\prime|+1}^{| W^\prime|+|L^\prime|}\frac{1}{i}-H( W^\prime)\\
& < & \sum_{i\in  W^\prime}v_{i}+\sum_{i=| W^\prime|+1}^{| W^\prime|+|L^\prime|} \frac{1}{i} -\sum_{i=| W^\prime|+1}^{| W^\prime|+|L^\prime|}\frac{1}{i}-H( W^\prime)\\
& = & \sum_{i\in  W^\prime}v_{i}-H( W^\prime)
\end{eqnarray*}
where the second inequality holds since $\sum_{i \in L^{\prime}} v_i < \frac{1}{n^{1-\delta}+1}+\cdots+\frac{1}{n^{1-\delta}+|L^\prime|}<\sum_{i=n^{1-\delta}+1}^{n^{1-\delta}+|L^\prime|} \frac{1}{i}\le\sum_{i=|W^\prime|+1}^{|W^\prime|+|L^\prime|} \frac{1}{i}$. Hence no player from $L$ is in $ALG$ nor in $ALG^{-i}$. Recall that we have already observed that for any $W^\prime\subset W$, the mechanism prefers the empty set over $W^\prime$ since $\sum_{i \in W^\prime} v_i <H(W^\prime)$.

We get that the mechanism outputs $W$, the only allocation that provides a non-negative value. Similarly, for every $i\in W$ it holds that $ALG^{-i}=\emptyset$. Therefore, the payment of every player $i\in W$ is:
\begin{eqnarray*} 
p_{i}&=&\left[ \sum_{j \in ALG^{-i}} v_j - H(ALG^{-i}) \right]-\left[\sum_{j\in W-\{i\}}v_{j} -H(W) \right]\\
& =&H(W)-\frac{|W|-1}{|W|}\cdot H(W) - \frac{|W|-1}{|W|}\cdot \varepsilon >\frac{H(W)}{|W|}-\frac{1}{n^3}
\end{eqnarray*}
Summing over all $i\in W$ we get that the sum of the payments is at least $H(W)-\frac{1}{n^2}$. By claim \ref{claim:VCG-Based-HS-HS-i>1/s}, $H(W)\ge \mathcal{H}_{|W|}$. Recall that $|W|\ge n^{(1-\delta)/2}$, therefore
$H(W)\ge\mathcal{H}_{n^{(1-\delta)/2}}\approx\frac{1-\delta}{2}\cdot\log{n}$. Hence, the sum of payments that the mechanism collects is at least $\frac{1-\delta}{2}\cdot\log{n}-1$.

We now analyze the social cost in this instance. Note that $\sum_{i\in L}v_{i}>\sum_{i=n^{1-\delta}+1}^{n}\frac{1}{i} - |L|\cdot\frac{1}{n^3}>\mathcal{H}_{n}-\mathcal{H}_{n^{1-\delta}} - 1\approx \log{n} - \log{n^{1-\delta}} - 1 =\delta\cdot\log{n} - 1 $
whereas the optimal allocation is $N$, with social cost $1$. The approximation ratio to the social cost is therefore no better than $\delta\cdot\log{n}-1$.
\end{proof}

\subsubsection*{Minimality of the Potential Function}\label{sub:Minimality-of-Potential}

The potential mechanism is an affine maximizer that uses the potential function as the function $H$ in the affine maximizer. Other choices for $H$ also guarantee that the mechanism will always cover the cost. However, we show that our choice of the potential function is not arbitrary in the sense that the potential function is the minimal function that always covers the cost.

\begin{theorem}
\label{prop:minimality-of-P}Let $\mathcal{A}$ be a VCG-based combinatorial cost-sharing mechanism for a domain that contains additive valuations and always covers the cost. I.e., $\mathcal{A}$ is a mechanism that outputs
\begin{equation}
\overrightarrow{ALG}=\arg\max_{\vec{S}\in 2^{M_1}\times \cdots \times 2^{M_n}}\sum_{i}v_{i}(S_{i})-H(\vec{S})\label{eq:VCG-affine-minimality}
\end{equation}
where $H:2^{M_{1}}\times \cdots \times2^{M_{n}}\rightarrow \mathbb R{\cup \{\infty\}}$ is normalized and monotone. Then, for every $\vec{S}\in2^{M_{1}}\times \cdots \times2^{M_{n}}$ it holds that $H(\vec{S})\ge P_{C}(\vec{S})$. 
\end{theorem}
\begin{proof}
We first prove for the case where $H(M)<\infty$. Suppose toward contradiction that the theorem is false and let $\vec{S}$ be a minimal allocation for which $H(\vec{S})<P_{C}(\vec{S})$ (i.e. for every strict subset $\vec{T}$, $H(\vec{T})\ge P_{C}(\vec{T})$). Define for each player $i$ the additive valuation $v_{i}(\{j\})=2H(\vec{M})$ if $j \in S_{i}$ and $v_{i}(\{j\})=0$ otherwise. Observe that $v_i(S_i)=2 \cdot |S_i|\cdot H(\vec{M})$.

Consider the allocation $\overrightarrow{ALG}=(ALG_{1},...,ALG_{n})$
that the mechanism outputs when the valuation profile is $(v_{1},...,v_{n})$.
It holds that $S_{i}\subseteq ALG_{i}$ for every $i$, otherwise the value of the allocation $\vec S+\{j\}$, for $j\notin  {ALG_i}$, $j\in S_i$, is higher in (\ref{eq:VCG-affine-minimality}). Thus, $\vec S\subseteq \overrightarrow{ALG}$ and by monotonicity $H(\overrightarrow{ALG})\geq H(\vec S) $. 

Let $Z=\{j|\exists i, j\notin S_i$, $j\in ALG_i\}$ be the set of items with $0$ value that the algorithm selects. Observe that $\Sigma_iv_i(S_i)=\Sigma_iv_i(ALG_i-Z)=\Sigma_iv_i(ALG_i)$ but $\overrightarrow{ALG}$ maximizes (\ref{eq:VCG-affine-minimality}). This implies that $H(\overrightarrow{ALG})\leq H(\vec S) $. Together we have that $H(\overrightarrow{ALG})= H(\vec S) $. Similar arguments give us that for every $i$, $\vec S-S_i\subseteq \overrightarrow{ALG^{-i}}$ and $H(\overrightarrow{ALG^{-i}})=H(\vec S-S_i)$.

Recall that $\vec{S}$ is a minimal allocation for which $H(\vec{S})<P_C(\vec{S})$, therefore for every $i$ it holds that $H(\vec{S}-S_{i})\ge P_{C}(\vec{S}-S_{i})$. Combining all together, the sum of payments is given by
\begin{eqnarray*}
\sum_{i\in N} p_{i} & = & \sum_{i\in N} \left [\left[ \sum_{j\neq i}v_j(ALG^{-i}_j) - H(\overrightarrow{ALG^{-i}}) \right] - \left[ \sum_{j\neq i } v_j(ALG_j) - H(\overrightarrow{ALG}) \right]\right ] \\
& = & \sum_{i \in N}\left [H(\overrightarrow{ALG})-H(\overrightarrow{ALG^{-i}})\right ]
 =  \sum_{i \in N}\left [H(\overrightarrow{S})-H(\vec S-S_i)\right ]\\ 
& < & \sum_{i\in N}\left [ P_{C}(\vec{S})-P_{C}(\overrightarrow{S}-S_{i}) \right ]\leq  C(\overrightarrow{S})
\end{eqnarray*}
Where the strict inequality follows due to our minimality assumption. Recall that $\vec S\subseteq \overrightarrow {ALG}$. By the monotonicity of the cost function $C(\vec S)\leq C(\overrightarrow{ALG})$. We get that $\Sigma_i\leq C(\overrightarrow{ALG})$, a contradiction.

Now suppose that $H(M)=\infty$. In this case, the theorem is obviously correct for $\vec S$ with $H(\vec S)=\infty>P_C(\vec S)$ so it remains to prove for $\vec S$ with $H(\vec S)<\infty$. Let $\vec T'$ be some maximal bundle with $H(\vec T)<\infty$ (i.e., for any $j\notin \vec T, H(\vec T+\{j\})=\infty$). Let $H'$ be the function that is defined only on items in $\vec T$ by $H'(\vec T')=H(\vec T)$ for every $\vec T'\subseteq T$. Observe that $H'(\vec T')<\infty$ and that the affine maximizer $\arg\max_{\vec S\subseteq \vec T}\Sigma_iv_i(S_i)-H'(S)$ coincides on subsets of $\vec T$ with (\ref{eq:VCG-affine-minimality}). Thus, applying the proof for finite $H(M)$ we get that for every $\vec T'\subseteq \vec T$, $H(\vec T')=H'(\vec T')\geq P_C(\vec T')$. The theorem follows.
\end{proof}

\section{Limitations of the Potential Mechanism}\label{appendix-limitations-potential}

In this section we revisit the potential mechanism and consider relaxing some of our assumptions on the
valuations and on the cost function. We first show that the potential
mechanism provides a bad approximation ratio to the social cost
for some non-subadditive functions. We then show that the potential mechanism is not budget balanced with a submodular cost function and simple submodular valuations such as unit demand.
\begin{proposition}
In the simple cost sharing setting, there is a non-subadditive cost function such that the Potential Mechanism provides an approximation ratio of $\Omega(n)$ to the social cost.
\end{proposition}
\begin{proof}
Consider $n$ players with $M_{i}=\{i\}$ for each $i\in N$. Let $C$ be the following non-subadditive
cost function: $C(\vec{S})=0$ for $\vec{S}\neq\vec{M}$, and $C(\vec{M})=1$.
Note that $P_{C}(\vec{S})=0$ for every $\vec{S}\neq\vec{M}$ and
$P_{C}(\vec{M})=\frac{1}{n}$ . Set $v_{i}=\frac{1}{n}+\varepsilon$
for every $i$ and for some small $\varepsilon>0$. The allocation
that maximizes (\ref{eq:mechanismVCG}) is $\overrightarrow{ALG}=\vec{M}$,
whereas the optimal allocation is $\overrightarrow{OPT}=\left(\{1\},\{2\},...,\{n-1\},\emptyset\right)$.
Therefore, the approximation ratio is $\frac{\pi(\overrightarrow{ALG})}{\pi(\overrightarrow{OPT})}=\frac{C(\overrightarrow{ALG})+\sum_{i\in}v_{i}(M_{i})-v_{i}(ALG_{i})}{C(\overrightarrow{OPT})+\sum_{i\in N}v_{i}(M_{i})-v_{i}(OPT_{i})}=\frac{1}{0+v_{n}(\{n\})}=\frac{1}{\frac{1}{n}+\varepsilon}=\frac{n}{1+n\varepsilon}$.
For $\varepsilon=\frac 1 n$, the approximation ratio of the mechanism is $\Omega(n)$.
\end{proof}

\begin{proposition}
The Potential Mechanism is $\Omega(n)$\emph{-budget-balanced} for
submodular cost function and unit demand valuations.
\end{proposition}
\begin{proof}
Consider $n$ players where for each $i$, $M_{i}=\{m_{1}^{i},..,m_{i-1}^{i},m_{i+1}^{i},...,m_{n}^{i}\}$.
The cost function $C:2^{M_{1}}\times\cdots\times2^{M_{n}}\to\{0,1,2\}$
is defined as follows: 
\begin{itemize} 
\item $C(\vec{S})=1$ if $\vec{S}\neq\vec{\emptyset}$ and there exists
$j$ such that $S_{j}=\emptyset$ and for every player $i\neq j$,
$S_{i}\subseteq\{m_{j}^{i}\}$. I.e., $n-1$ players ``agree'' on an index $j$ and either take the corresponding item or nothing, and player $j$ gets nothing. 
\item $C(\emptyset)=0$. In any other case, $C(\vec{S})=2$.
\end{itemize}
Clearly, $C$ is monotone. One can verify that it is also submodular.


Consider the following valuation profile: every player is unit
demand with $v_{i}(S_{i})=L$ for every $\emptyset\neq S_{i}\subseteq M_{i}$,
where $L>2\cdot\mathcal{H}_{n}$, and $v_{i}(\emptyset)=0$. 

By proposition \ref{prop:P<HnC}, for submodular function $C$ it holds
that $P_{C}(\vec{S})\le\mathcal{H}_{n}\cdot C(\vec{S})$ for every
$\vec{S}$, therefore $P_{C}(\vec{S})\le2\cdot\mathcal{H}_{n}<L$.
Hence, in every allocation $\vec{S}$ that maximizes (\ref{eq:mechanismVCG})
every player receives at least one service. 
By monotonicity of $C$, $P_{C}$ is monotone as well (proposition
\ref{prop:C-submodular->PSubmodular}), therefore there is an allocation
that maximizes (\ref{eq:mechanismVCG}) where no player gets more than one
service. Our goal now is to find the allocation that
allocates every player one service and minimizes the potential function
$P_{C}(\vec{S})$.

Recall that by \cite{hart1989potential} and Section \ref{sec:function}, given an allocation $\vec S= (S_1,\ldots, S_n)$, the marginal value $P_C(\vec S)-P_C(\vec S-S_i)$ equals the Shapley value\footnote{Recall that the Shapley value of a player is its expected marginal contribution to the cost in a random permutation.} of player $i$ in the cooperative game with $N$ players and cost function $C'$, where the cost of serving a set $T$ is $C'(T)=C(\cup_{i\in T}S_i)$. For each player $i$ let $Shapley_i(\{1,...,i\})$ denote the Shapley value of player $i$ in the coalition $\{1,...,i\}$. We get that $P_C(\vec S)=\Sigma_{i=1}^nShapley_i(\{1,...,i\})$.



Fix some allocation $\vec{S}$ where each player gets one service. For every $i$, we have that $C'(\{i\})= 1$. Therefore $Shapley_i(\{1,...,i\})\geq 1/i$. When all players agree on the index we furthermore have that $Shapley_i(\{1,...,i\})= 1/i$, since only the first player in the permutation has a positive marginal contribution to the cost.

If not all players agree on the permutation, then the marginal contribution to the cost of the first player who disagrees on the index is $1$. Therefore, both this player and the first player has a marginal contribution of $1$. In this case we thus have $Shapley_i(\{1,...,i\})> 1/i$. Hence, the allocation in which every player gets at least one service that minimizes the potential function is the one in which $n-1$ players agree on an index and one player does not agree, e.g., $\overrightarrow{ALG}=\left\{\left\{m_n^1\right\},\left\{m_n^2\right\},\ldots,\left\{m_{n}^{n-1}\right\},\left\{m_1^n\right\}\right\}$. We note that $P_C(\overrightarrow{ALG})=\sum_{i=1}^n Shapley_i(\{1,2,...,i\}) = \sum_{i=1}^{n-1} \frac 1 i + 1 = 1+\mathcal{H}_{n-1}$. Similarly, for every $i$ the allocation that minimizes the potential function $P_C$ when player $i$ does not get a service and any other player gets exactly one service is $\overrightarrow{ ALG^{-i}}=\left\{ \left\{ m_{i}^{1}\right\} ,\left\{ m_{i}^{2}\right\} ,...,\left\{ m_{i}^{i-1}\right\} ,\emptyset,\left\{ m_{i}^{i+1}\right\} ,...,\left\{ m_{i}^{n}\right\} \right\}$, for which $P_C(\overrightarrow{ALG^{-i}})=\sum_{i=1}^{n-1} \frac{1}{i}=\mathcal{H}_{n-1}$.

Therefore, the payment of player $i$ is given by: 
\begin{eqnarray*}
p_{i}&=&\sum_{j\in N\backslash\{i\}}v_{i}(ALG_{j}^{-i})-P_{C}(\overrightarrow{ALG^{-i}})-\left[\sum_{j\in N\backslash\{i\}}v_{j}(ALG_{j})- P_{C}(\overrightarrow{ALG})\right] \\
&=&(n-1)\cdot L - P_{C}(\overrightarrow{ALG^{-i}}) - \left[ (n-1)\cdot L - P_C(\overrightarrow{ALG}) \right] \\ &=& P_{C}(\overrightarrow{ALG})-P_{C}(\overrightarrow{ALG^{-i}})=\mathcal{H}_{n-1} + 1 - \mathcal{H}_{n-1}=1
\end{eqnarray*}
In total, the sum of payments is $n$, whereas the cost is $2$. Therefore, the potential mechanism is at least $\frac{n}{2}$-budget-balanced in this case.
\end{proof}

\section{Groupstrategyproof Mechanisms}\label{sec-gsp}

Recall that VCG-based mechanisms, such as the potential mechanism,
are not group-strategyproof. That is, two or more players can collude and increase their profit. For example, consider the potential mechanism
and the excludable public good problem. Note that two players with
values $v_{1},v_{2}=\frac{1}{2}+\varepsilon$ do not get the service,
since $P_{C}(\{1\})=P_{C}(\{2\})=1$ and $P_{C}(\{1,2\})=1.5$. However,
if they both misreport $v_{1},v_{2}=1+\varepsilon$, they
served and each pays $p_{i}=\frac{1}{2}$, which provides a profit of
$\varepsilon>0$ for both of them.

In the single parameter setting, many papers study groupstrategyproof
mechanisms (e.g., \cite{bleischwitz2008group,mehta2007beyond,moulin1999incremental,moulin2001strategyproof}).
Constructing groupstrategyproof mechanisms for a multi-parameter setting is even more
challenging. In this section we present and analyze some simple groupstrategyproof
mechanisms for combinatorial cost sharing. We start with recalling some standard definitions:

\emph{Weakly groupstrategyproofness} requires that no subset of players can report a different valuation, and strictly increase the profit for every player in the subset:
\begin{itemize}
\item A combinatorial cost-sharing mechanism is \emph{weakly groupstrategyproof
}(WGSP) if for every subset of players $K\subseteq N$ and every valuation
profile $\vec{v}^{\prime}$ it holds that if $v_{i}(S_{i})-p_{i}\le v_{i}(S_{i}^{\prime})-p'_{i}$
for every $i\in K$ then $v_{j}(S_{j})-p_{j}=v_{j}(S_{j}^{\prime})-p_{j}^{\prime}$
for some $j\in K$, where $(\vec{S},\vec{p})$ and $(\vec{S}^{\prime},\vec{p}^{\prime})$
are the outputs of the mechanism for $\vec{v}$ and for the instance that is obtained from $\vec v$ by setting $v_i=v'_i$ for every $i\in K$, respectively.
\end{itemize}
The strongest notion is \emph{groupstrategyproofness} in which every misreporting by a subset of players which strictly
increases the profit of one of the members, strictly decreases the
profit of another member:
\begin{itemize}
\item A combinatorial cost-sharing mechanism is \emph{groupstrategyproof
}(GSP) if for every subset of players $K\subseteq N$ and every valuation
profile $\vec{v}^{\prime}$ it holds that if $v_{i}(S_{i})-p_{i}\le v_{i}(S_{i}^{\prime})-p_{i}^{\prime}$
for every $i\in K$ then $v_{i}(S_{i})-p_{i}=v_{i}(S_{i}^{\prime})-p_{i}^{\prime}$
for every $i\in K$, where $(\vec{S},\vec{p})$ and $(\vec{S}^{\prime},\vec{p}^{\prime})$
are the outputs of the mechanism for $\vec{v}$ and for the instance that is obtained from $\vec v$ by setting $v_i=v'_i$ for every $i\in K$, respectively. 
\end{itemize}
Note that every mechanism that is GSP is also WGSP, and that every
mechanism that is WGSP is also truthful.

\subsection{The Sequential Mechanism}

Here we restate the\emph{ Sequential Mechanism }\cite{moulin1999incremental}
and identify two cases in which it is groupstrategyproof,
\emph{1-}budget-balance, and provides an approximation ratio of
\emph{n} to the social cost. Similarly to Section \ref{sec:The-Mechanism}, the sequential
mechanism provides the mentioned guarantees when the cost function is submodular and the valuations are supermodular
and when the cost function is player-wise symmetric and the valuations are symmetric as well.

\subsection*{The Sequential Mechanism}
\begin{enumerate}
\item Set $ALG_i = \emptyset$ for every player $i\in N$.
\item For player $i=1$ to $n$:

\begin{enumerate}
\item Let $\mathcal{A}_{i}=\left\{ A_{i}:A_{i}\in\arg\max_{S_{i}\subseteq M_{i}}v_{i}(S_{i})-\left[C(\overrightarrow{ALG}\cup S_{i})-C(\overrightarrow{ALG})\right]\right\} $.
\item \label{enu:Allocate-to-player}Allocate to player $i$ a bundle $ALG_{i}\in\mathcal{A}_i$
with the maximal size in $\mathcal{A}_{i}$. 
\item Charge player $i$:  $p_{i}=C(ALG_1,ALG_2,\ldots,ALG_{i-1},ALG_i)-C(ALG_1,ALG_2,\ldots,ALG_{i-1})$.
\end{enumerate}
\end{enumerate}
\begin{theorem}
\label{thm:SeqMec}
Suppose that either:
\begin{enumerate}
\item \label{enu:Condition1Seq}The cost function $C$ is submodular and
the valuation functions are supermodular.
\item \label{enu:Condition2Seq}The cost function $C$ is submodular and
player-wise symmetric and the valuation functions are symmetric.
\end{enumerate}
Then the Sequential Mechanism is group-strategyproof, 1-budget-balanced,
provides an $n$-approximation to the social cost and satisfies IR
and NPT.
\end{theorem}

\begin{lemma}
\label{lem:seq-NPT-IR-BB}For every monotone and normalized cost function
and normalized valuation function the sequential Mechanism is 1-budget-balanced
and satisfies IR and NPT.
\end{lemma}
\begin{proof}
The mechanism satisfies individual rationality since a player can
choose the empty bundle with zero payment, hence he never pays more
than the value of the bundle allocated to him. By the monotonicity
of the cost function, the payments are always positive. Moreover,
since $C$ is normalized, the sum of payments is $C(\overrightarrow{ALG})$.
\end{proof}

In order to prove groupstrategyproofness, we first prove claim \ref{union-of-maximum-is-maximum}
which states that for the combinatorial cost sharing with a submodular
cost function and supermodular valuations, the union of bundles which
maximize the profit also maximizes the profit.
\begin{claim}\label{union-of-maximum-is-maximum}
Assume that the cost function $C$ is submodular. Consider player $i$ with supermodular valuation $v_i$, and consider an allocation $\overrightarrow{ALG}$ with $ALG_j=\emptyset$ for every $j\geq i$. For each bundle $S\subseteq M_i$, let $p_{S}$ denote the price of the bundle which is its marginal cost, i.e., $p_{S}=C(\overrightarrow{ALG}\cup S)- C(\overrightarrow{ALG})$. Let $S_{1},S_{2}$
be two bundles that maximize player $i$'s profit, i.e., $S_{1},S_{2}\in\arg\max_{S}v_i(S)-p_{S}$.
Then $S_{1}\cup S_{2}$ maximizes the profit as well, that is, $S_{1}\cup S_{2}\in\arg\max_{S}v_i(S)-p_{S}$.
\end{claim}
\begin{proof}
Let $T=S_{1}\cap S_{2}$ and denote by $U_{1}$ and $U_{2}$ the unique
elements in each of the bundles: $U_{1}=S_{1}\backslash T$, $U_{2}=S_{2}\backslash T$.
The following holds:

\begin{eqnarray*}
v_{i}(S_{1}\cup U_{2})-v_{i}(S_{1}) & = & v_{i}(U_{1}\cup T\cup U_{2})-v_{i}(U_{1}\cup T)\\
 & \ge & v_{i}(T\cup U_{2})-v_{i}(T)\\
 & \ge & \left[C(\overrightarrow{ALG}\cup T\cup U_{2})-C(\overrightarrow{ALG})\right]-\left[C(\overrightarrow{ALG}\cup T)-C(\overrightarrow{ALG})\right]\\
 & \ge & \left[C(\overrightarrow{ALG}\cup T\cup U_{1}\cup U_{2})-C(\overrightarrow{ALG})\right]-\left[C(\overrightarrow{ALG}\cup T\cup U_{1})-C(\overrightarrow{ALG})\right]\\
 & = & \left[C(\overrightarrow{ALG}\cup S_{1}\cup U_{2})-C(\overrightarrow{ALG})\right]-\left[C(\overrightarrow{ALG}\cup S_{1})-C(\overrightarrow{ALG})\right]
\end{eqnarray*}

The first inequality is due to supermodularity of $v_i$. The second
inequality is by the fact that $S_{2}$ maximizes $v(S)-p_{S}$, therefore
$v(S_{2})-\left[C\left(\overrightarrow{ALG}\cup S_{2}\right)-C(\overrightarrow{ALG})\right]\ge v_{i}(T)-\left[C(\overrightarrow{ALG}\cup T)-C(\overrightarrow{ALG})\right]$.
The third inequality is by submodularity of $C$. Note that $S_{1}\cup U_{2}=S_{1}\cup S_{2}$,
hence we conclude that 
\begin{eqnarray*}
v(S_{1}\cup S_{2})-\left[C(\overrightarrow{ALG}\cup S_{1}\cup S_{2})-C(\overrightarrow{ALG})\right] & \ge & v(S_{1})-\left[C(\overrightarrow{ALG}\cup S_{1})-C(\overrightarrow{ALG})\right]
\end{eqnarray*}
therefore $S_{1}\cup S_{2}$ maximizes the profit as well.\end{proof}
\begin{lemma}
In each one of the settings (\ref{enu:Condition1Seq}) and 
(\ref{enu:Condition2Seq}) of theorem \ref{thm:SeqMec} the sequential mechanism is groupstrategyproof. 
\end{lemma}
\begin{proof}
We denote by $\overrightarrow{ALG}=(ALG_{1},\ldots,ALG_{n})$ and by $p_{1},\ldots,p_{n}$
the output of the mechanism given the valuation profile $\vec{v}$.
Consider player $i$. Note that nothing that players $i+1,...,n$
do affects player $i$'s prices. Hence, a strategic player $i$ chooses
a bundle that maximizes his profit (which depends only on $ALG_{1},...,ALG_{i-1}$).
In case of tie-breaking, he chooses the bundle that decreases the
prices for players $i+1,...,n$, if such exists. We prove by
induction on the players that for every $i$, $ALG_{i}$ is the bundle
that both maximizes his profit and always minimizes players $i+1,...,n$
prices.

Consider player $1$, and assume that she chooses a bundle $A_{1}^{\prime}\neq ALG_{1}$
that maximizes her profit. Note that in setting (\ref{enu:Condition1Seq}) of theorem \ref{thm:SeqMec}, claim \ref{union-of-maximum-is-maximum} implies $A'_{1}\subset ALG_{1}$. Fix some choice of actions by
the rest of the players given that player 1 chooses $A_{1}^{\prime}:$
$(A_{2}^{\prime},p_{2}^{\prime}),...,(A_{n}^{\prime},p_{n}^{\prime})$
, i.e. player $j$ chooses $A_{j}^{\prime}$ and pays $p_{j}^{\prime}$.
Since $C$ is submodular, for every player $j>1$ it holds that 
\[
C(A'_{1},A'_{2}\ldots,A'_{j-1},A'_{j})-C(A'_{1},A'_{2}\ldots,A'_{j-1})\ge C(ALG_{1},A'_{2},\ldots, A'_{j-1}, A'_{j})-C(ALG_{1}, A'_{2},\ldots,A'_{j-1})
\]
i.e., every player $j>1$ pays no more than $p_{j}^{\prime}$ assuming
that player 1 chooses the bundle $ALG_{1}$. Recall that $ALG_{1}$ maximizes
player 1's profit, therefore he does not join any coalition, and takes $ALG_{1}$. 

We note that the proof is similar in setting (\ref{enu:Condition2Seq})
of theorem \ref{thm:SeqMec}: in this case $|A'_{1}|\le|ALG_{1}|$.
Without loss of generality we can assume that $A'_{1}\subseteq ALG_{1}$ since for
every $\tilde{A}_{1}\subseteq ALG_{1}$ with $|A'_{1}|=|\tilde{A}_{1}|$
it holds that $v_{1}(A'_{1})=v_{1}(\tilde{A}_{1})$ and for
every allocation of other players $\vec{S}_{-1}=\left(S_{2},...,S_{n}\right)$
it holds that $C(A_{1}^{\prime},\vec{S}_{-1})=C(\tilde{A}_{1},\vec{S}_{-1})$. 

The proof of the inductive step for every player $i>1$ is similar to that of player $1$.
\end{proof}
\begin{lemma}
\label{lem:SeqMec-napprox}Assume that the cost function $C$ is subadditive
(in particular, submodular). Then for every valuation profile the
sequential mechanism guarantees an approximation ratio of $n$.
\end{lemma}
\begin{proof}
Consider a valuation profile $\vec{v}$ and let $\overrightarrow{ALG}=\left(ALG_{1},...,ALG_{n}\right)$
be the output. Let $\overrightarrow{OPT}=(OPT_{1},...,OPT_{n})$ be the
allocation that minimizes the social cost. For every
player $i$ it holds that: 
\begin{eqnarray*}
 &  & v_{i}(ALG_{i})-\left[C(ALG_{1},\ldots,ALG_{i-1},ALG_{i})-C(ALG_{1},\ldots,ALG_{i-1})\right]\\
 & \ge & v_{i}(OPT_{i})-\left[C(ALG_{1},\ldots,ALG_{i-1},OPT_{i})-C(ALG_{1},\ldots,ALG_{i-1})\right]\\
 & \ge & v_{i}(OPT_{i})-C(OPT_{i})\\
 & \ge & v_{i}(OPT_{i})-C(\overrightarrow{OPT})
\end{eqnarray*}
where the first inequality holds since $ALG_{i}$ is the bundle that
maximizes the profit of player $i$ in the $i^{th}$ iteration of
the mechanism, the second inequality is due to the subadditivity of
$C$ and the last inequality is due to monotonicity of $C$. Summing
over all players we get: 
\[
\sum_{i\in N}v_{i}(ALG_{i})-C(\overrightarrow{ALG})\ge\sum_{i\in N}v_{i}(OPT_{i})-n\cdot C(\overrightarrow{OPT})
\]
The approximation ratio is therefore:
\[
\frac{\pi(\overrightarrow{ALG})}{\pi(\overrightarrow{OPT})}=\frac{C(\overrightarrow{ALG})+\sum_{i\in N}\left[v_{i}(M_{i})-v_{i}(ALG_{i})\right]}{C(\overrightarrow{OPT})+\sum_{i\in N}\left[v_{i}(M_{i})-v_{i}(OPT_{i})\right]}\le\frac{n\cdot C(\overrightarrow{OPT})+\sum_{i\in N}\left[v_{i}(M_{i})-v_{i}(OPT_{i})\right]}{C(\overrightarrow{OPT})+\sum_{i\in N}\left[v_{i}(M_{i})-v_{i}(OPT_{i})\right]}\le n
\]
\end{proof}

To see that the analysis is tight, consider the excludable public
good problem: $M_{i}=\{i\}$ for every $i\in N$. For every $\vec{S}\neq\vec{\emptyset}$,
$C(\vec{S})=1$. For every player $i$, set $v_{i}(M_{i})=1-\varepsilon$
for arbitrary small $\varepsilon>0$. Then the mechanism outputs $\overrightarrow{ALG}=\left(\emptyset,...,\emptyset\right)$,
and the social cost is $\pi(\overrightarrow{ALG})=0+\sum_{i}v_{i}=n-n\cdot\varepsilon$.
The optimal allocation is $\vec{M}$ and its social cost is $\pi(\vec{M})=C(\vec{M})=1$.
Since $\varepsilon$ is arbitrary small, the approximation ratio is
no better than $n$.

\subsection{A Weakly Groupstrategyproof Mechanism }

Recall the weaker notion of groupstrategyproofness - weakly-groupstrategyproofness
(WGSP), in which there is no coalition which can get strictly better
profit for every player in the coalition by misreporting valuations.
Relaxing the requirement for groupstrategyproofnes, and requiring
WGSP instead, allow us to use the sequential mechanism on a wider
family of combinatorial cost-sharing problems and get the same guarantees.

\subsection*{The Sequential Mechanism for Subadditive Cost and General Valuations}
\begin{enumerate}
\item Set $ALG_i = \emptyset$ for every player $i\in N$.
\item For player $i=1$ to $n$:

\begin{enumerate}
\item Allocate to player $i$ some bundle $ALG_{i}\in\arg\max_{S_{i}\subseteq M_{i}}v_{i}(S_{i})-\left[C(\overrightarrow{ALG}\cup S_{i})-C(\overrightarrow{ALG})\right]$.

If there are several bundles that maximize the profit, choose the
lexicographically first one.

\item Charge player $i$: $p_{i}=C(ALG_1,ALG_2,\ldots,ALG_{i-1},ALG_i)-C(ALG_1,ALG_2,\ldots,ALG_{i-1})$.
\end{enumerate}
\end{enumerate}
\begin{theorem}
Suppose that the cost function $C$ is subadditive. Then, the Sequential
Mechanism is weakly-groupstrategyproof, 1-budget-balanced, and provides an approximation ratio of $n$ to the social cost.
\end{theorem}
\begin{proof}
All proofs remain the same except weakly groupstrategyproofness. Recall that a coalition might deviate and choose different bundles if those bundles strictly
increase the profit of every player in the coalition. We note that
for every player $j\ge1$, nothing that players $j+1,...,n$ do affects
player $j$'s profit, therefore he does not deviate and chooses the
bundle $ALG_{j}$ which is a bundle that maximize his profit.
\end{proof}

\bibliographystyle{plain}
\bibliography{combinatorialcost}

\appendix

\section{Missing Proofs}

\subsection{Proof of Proposition \ref{prop:uniquePc}}\label{appendix:ExplicitProof}

	Starting with $P_{C}(\emptyset)=0$, the potential function $P_{C}(\vec{S})$
	is uniquely determined by the recursion $P_{C}(\vec{S})=\frac{C(\vec{S})+\sum_{i\in N;S_{i}\neq\emptyset}P_{C}(\vec{S}-S_{i})}{n_{S}}$,
	where we denote by $n_{S}$ the number of non-empty sets $S_{i}$
	of $\vec{S}$. 
	
	Next, we show that $P_{C}(\vec S)=\sum_{I\subseteq N}\frac{C\left(\cup_iS_i\right)}{|I|\cdot{n \choose \left|I\right|}}$
	satisfies the requirements of the potential function. Clearly, since
	$C(\emptyset)=0$ it holds that $P_{C}(\emptyset)=0$. Notice that for every $i\in N$ we can rewrite $P_{C}(\vec{S})$:
	\begin{eqnarray*}
		P_{C}(S) & = & \sum_{I\subseteq N}\frac{C\left(\cup_iS_i\right)}{|I|\cdot{n \choose |I|}}=\sum_{I\subseteq N-\{i\}}\left[\frac{C\left(\cup_{j\in I}S_j\cup S_{i}\right)}{(|I|+1)\cdot{n \choose |I|+1}}+\frac{C\left(\cup_{j\in I}S_j\right)}{|I|\cdot{n \choose |I|}}\right]
	\end{eqnarray*}
	Consider the sum of discrete gradients: 
	\begin{eqnarray}
	 & \sum_{i\in N}&\left[P_{C}(\vec{S})-P_{C}(\vec{S}-S_{i})\right]\nonumber \\
	& = & \sum_{i\in N}\left[\sum_{I\subseteq N-\{i\}}\left[\frac{C\left(\cup_{j\in I}S_j\cup S_{i}\right)}{(|I|+1)\cdot{n \choose |I|+1}}+\frac{C\left(\cup_{j\in I}S_j\right)}{|I|\cdot{n \choose |I|}}\right]-\sum_{I\subseteq N-\{i\}}\left[\frac{C\left(\cup_{j\in I}S_j\cup\emptyset\right)}{(|I|+1)\cdot{n \choose |I|+1}}+\frac{C\left(\cup_{j\in I}S_j\right)}{\left(|I|\right)\cdot{n \choose |I|}}\right]\right]\nonumber \\
	& = & \sum_{i\in N}\left[\sum_{I\subseteq N-\{i\}}\frac{C\left(\cup_{j\in I}S_j\cup S_{i}\right)}{(|I|+1)\cdot{n \choose |I|+1}}-\frac{C\left(\cup_{j\in I}S_j\right)}{(|I|+1)\cdot{n \choose |I|+1}}\right]\label{eq-appendix:double-sum}
	\end{eqnarray}
	By re-arranging the above double sum, we obtain constants $\left\{ a_{I}\right\} _{I\subseteq N}$
	independent of $S$ such that 
	\[
	\sum_{i\in N}\left[P_{C}(\vec{S})-P_{C}(\vec{S}-S_{i})\right]=\sum_{I\subseteq N}a_{I}\cdot C\left(\cup_{j\in I}S_j\right)
	\]
Since $C(\vec{\emptyset})=0$, it suffices to show that $a_{N}=1$ and that for every non-empty $I\neq N$, $a_{I}=0$. Let $I\subseteq N$, $I\neq\emptyset$, observe that the summand $C\left(\cup_{j\in I}S_j\right)$
	appears in (\ref{eq-appendix:double-sum}) exactly $|I|$ times with a plus
	sign, once for every $i\in I$. In each of these times the coefficient
	of $C\left(\cup_{j\in I}S_j\right)$ is $\frac{1}{|I|\cdot{n \choose |I|}}$.
	If $I\neq N$, $C\left(\cup_{j\in I}S_j\right)$ appears
	in (\ref{eq-appendix:double-sum}) exactly $n-|I|$ times with a minus sign,
	once for every $i\notin I$. In each time the coefficient of $-C\left(\cup_{j\in I}S_j\right)$
	is $\frac{1}{(|I|+1)\cdot{n \choose |I|+1}}$. Summing the positive
	and negative coefficients yields 
	\begin{eqnarray*}
		a_{I}=\frac{|I|}{|I|\cdot{n \choose |I|}}-\frac{n-|I|}{(|I|+1)\cdot{n \choose |I|+1}} & = & \frac{1}{{n \choose |I|}}-\frac{(n-|I|)\cdot(|I|+1)!\cdot(n-|I|-1)!}{(|I|+1)\cdot n!}\\
		& = & \frac{1}{{n \choose |I|}}-\frac{(n-|I|)\cdot|I|!\cdot(n-|I|-1)!}{n!}\\
		& = & \frac{1}{{n \choose |I|}}-\frac{1}{{n \choose |I|}}=0
	\end{eqnarray*}
	To conclude the proof, we note that $C\left(\cup_{j\in N}S_j\right)$ does
	not appear in (\ref{eq-appendix:double-sum}) with a minus sign. Thus, a similar
	observation shows that 
	\[
	a_{N}=\frac{|N|}{|N|\cdot{n \choose |N|}}=1
	\]

\end{document}